\documentclass[11pt]{amsart}
\usepackage{graphicx}
\usepackage{amscd}
\usepackage{amsmath}
\usepackage{amsfonts}
\usepackage{amssymb}
\usepackage{setspace}
\usepackage{enumerate}         
\usepackage{color}
\usepackage{url}
\usepackage{amsthm}
\usepackage{hyperref}
\usepackage{bm}
\usepackage{xy}
\usepackage{color}
\allowdisplaybreaks[4]

\usepackage{geometry}
\theoremstyle{plain}
\newtheorem{theorem}{Theorem}[section]

\newtheorem{corollary}[theorem]{Corollary}

\newtheorem{lemma}[theorem]{Lemma}

\newtheorem{problem}[theorem]{Problem}
\newtheorem{proposition}[theorem]{Proposition}

\newtheorem{definition}[theorem]{Definition}

\newtheorem{assumption}[theorem]{Assumption}
\newtheorem{requirement}[theorem]{Requirement}
\theoremstyle{remark}
\newtheorem{remark}[theorem]{Remark}

\numberwithin{equation}{section}

\newcommand{\ind}{\mathbf{1}}
\newcommand{\rsto}{]\!\kern-1.8pt ]}
\newcommand{\lsto}{[\!\kern-1.7pt [}

\numberwithin{equation}{section}

\newcommand{\RR}{\mathbb{R}}
\newcommand{\QQ}{\mathbb{Q}}
\newcommand{\QQtilde}{\widetilde{\mathbb{Q}}}
\newcommand{\PP}{\mathbb{P}}
\newcommand{\CC}{\mathbb{C}}

\newcommand{\EE}{\mathbb{E}}

\newcommand{\cF}{\mathcal{F}}
\newcommand{\cH}{\mathcal{H}}
\newcommand{\cG}{\mathcal{G}}

\newcommand{\cT}{\mathcal{T}}
\newcommand{\cX}{\mathcal{X}}

\newcommand{\mfU}{\mathfrak{U}}
\newcommand{\Yhat}{\hat{Y}}
\newcommand{\xihat}{\hat{\xi}}
\newcommand{\Hline}{\bar{H}}

\newcommand{\cD}{\mathcal{D}}

\newcommand{\Excond}[3]{\mathbb{E}^{#1}\left[\left.#2\right|#3\right]}  

\newcommand{\im}{\ensuremath{\mathsf{i}}}

\makeatletter
\newcommand{\subjclassname@JEL}{JEL Classification}
\makeatother

\begin{document}
\title[A general HJM framework for multiple curve modeling]{A general HJM framework for multiple yield curve modeling}

\author{Christa Cuchiero}
\address[Christa Cuchiero]{Vienna University of Technology  \newline
\indent Financial and Actuarial Mathematics,\newline%
\indent Wiedner Hauptstrasse 8-10/105-1 1040 Wien, Austria}
\email[Christa Cuchiero]{cuchiero@fam.tuwien.ac.at}%

\author{Claudio Fontana}
\address[Claudio Fontana]{Laboratoire de Probabilit\'es et Mod\`eles Al\'eatoires, Universit\'{e} Paris Diderot, \newline
\indent avenue de France, F-75205 Paris, France}
\email[Claudio Fontana]{fontana@math.univ-paris-diderot.fr}%
\thanks{The research of C.F. was partly supported by a Marie Curie Intra European Fellowship within the 7th European Community Framework Programme under grant agreement PIEF-GA-2012-332345.}

\author{Alessandro Gnoatto}
\address[Alessandro Gnoatto]{Mathematisches Institut der LMU M\"unchen,\newline%
\indent Theresienstrasse, 39 D-80333 M\"unchen}
\email[Alessandro Gnoatto]{gnoatto@mathematik.uni-muenchen.de}%

\begin{abstract}
We propose a general framework for modeling multiple yield curves which have emerged after the last financial crisis. In a general semimartingale setting, we provide an HJM approach to model the term structure of multiplicative spreads between FRA rates and simply compounded OIS risk-free forward rates.
We derive an HJM drift and consistency condition ensuring absence of arbitrage and, in addition, we show how to construct models such that multiplicative spreads are greater than one and ordered with respect to the tenor's length.
When the driving semimartingale is an affine process, we obtain a flexible and tractable Markovian structure.
Finally, we show that the proposed framework allows to unify and extend several recent approaches to multiple yield curve modeling. 
\end{abstract}


\maketitle

\section{Introduction}\label{sec:Intro}

The last financial crisis has profoundly affected fixed income markets. Most notably, significant spreads have emerged between interbank (Libor/Euribor) rates and (risk-free) OIS rates as well as between interbank rates associated to different tenor lengths, mainly due to an increase in credit and liquidity risk. While negligible in the pre-crisis environment, such spreads represent nowadays one of the most striking features of interest rate markets, with the consequence that interbank rates cannot be considered risk-free any longer (see Section \ref{subsec:market} for more details). From a modeling perspective, this new market situation necessitates a new generation of interest rate models, which are able to represent in a consistent way the evolution of multiple yield curves and allow to value fixed income derivatives.

The present paper aims at providing a coherent and general modeling approach for multiple interest rate curves. We shall adopt an HJM framework driven by general semimartingales in the spirit of~\cite{KK13} in order to model the joint evolution of the term structure of OIS zero coupon bond prices and of the term structure of spreads between forward rates linked to interbank rates and OIS forward rates. More specifically, we shall model the term structure of \emph{multiplicative spreads} between (normalized) forward rates implied by market forward rate agreement (FRA) rates, associated to a family of different tenors, and (normalized) simply compounded OIS forward rates. 
Besides admitting a natural economic interpretation in terms of forward exchange premiums, multiplicative spreads provide a particularly convenient parametrization of the term structures of interbank rates.
Referring to Section \ref{subsec:model} for a more detailed discussion of the proposed framework, let us just mention here that, additionally to the great generality and flexibility, this modeling approach has the advantage of considering as model fundamentals easily observable market quantities and of leading to a clear characterization of order relations between spreads associated to different tenors. Moreover, specifying the driving semimartingale as an affine process leads to a Markovian structure and tractable valuation formulas.

By adopting an abstract HJM formulation, we derive a simple HJM drift and consistency condition ensuring absence of arbitrage in a general semimartingale setting. Moreover, starting from a given tuple of basic building blocks, we provide a general construction of arbitrage-free multiple yield curve models such that spreads are ordered and greater than one. To this effect, we prove existence and uniqueness of the SPDEs associated to the forward curves when translated to the Musiela parametrization and show how to guarantee the consistency condition by constructing an appropriate pure jump process whose compensator solves a random generalized moment problem. As shown in Section~\ref{sec:relations}, most of the multiple curve models proposed in the literature can be recovered as suitable specifications of our general framework, thus underlying the high flexibility of the proposed approach.

The multiple curve phenomenon has attracted significant attention from market practice as well as from the academic literature (see, e.g., the recent book~\cite{Henr14} and the references therein). 
To the best of our knowledge, the first paper highlighting the relevance of the multiple curve issue shortly before the beginning of the credit crunch was \cite{hen07}. 
From a modeling perspective, as in the case of classical interest rate models, most of the models proposed so far in the literature can be ascribed to three main mutually related families: short-rate approaches, Libor market models and HJM models. Referring to Section~\ref{sec:relations} for a detailed comparison of the different approaches, we just mention that multiple curve short rate models have been first introduced in \cite{kitawo09}, \cite{ken10}, \cite{fitr12} and, more recently, in \cite{MR14} and \cite{GM:14}, while Libor market models have been extended to the multiple curve setting in \cite{mer10b}, \cite{mer10} and, more recently, in \cite{GPSS14}. 
In a related context, \cite{merxie12} propose a model for additive spreads which can be applied on top of any classical single-curve interest rate model.
Our approach is closer to the multiple curve HJM-type models proposed in the literature, see in particular \cite{mopa10}, \cite{pata10}, \cite{fushita09}, \cite{cre12} and \cite{CGNS:13} (note also that the idea of modeling multiplicative spreads goes back to~\cite{hen07} and \cite{hen10}).
We also want to mention that the joint modeling of the risk-free term-structure together with a ``risky'' term-structure goes back to the earlier contributions \cite{JT:95} and \cite{DJ:02}. 
More recently, spreads between Libor rates and risk-free rates have also been modeled in \cite{grba12} by introducing default risk in a 
Libor market model.


The paper is organized as follows. Section~\ref{subsec:market} introduces the basic quantities considered in the paper and explains the philosophy behind the proposed modeling approach.
In Section~\ref{genFram}, we define a general HJM-type framework, inspired by~\cite{KK13},  which we then apply to multiple yield curve modeling. In particular, we derive a drift and consistency condition which ensures absence of arbitrage in general HJM models driven by It\^o-semimartingales.
In Section~\ref{sec:construction} we show how to construct arbitrage-free models with ordered spreads satisfying the drift and consistency condition. 
In Section~\ref{sec:implementation} we illustrate the main aspects related to the implementation of the proposed framework and provide general guidelines to model calibration. Moreover, we present general valuation formulas and introduce an especially tractable specification based on affine processes. 
In Section~\ref{sec:relations} we show how most of the existing multiple curve models can be easily embedded in our framework. Finally, Appendix~\ref{sec:FRArates} contains a review of pricing under collateral and its implication for the definition of fair FRA rates, Appendix~\ref{app:FX} illustrates a foreign exchange analogy, Appendix~\ref{appendix:local_ind} briefly recalls the notion of local independence of semimartingales and Appendix~\ref{app:proof} contains the technical proofs of several results of Section~\ref{sec:construction}.

\section{Modeling the post-crisis interest rate market}	\label{subsec:market}

In fixed income markets, the underlying quantities of the vast majority of traded contracts are Libor (or Euribor) rates $L_T(T,T+\delta)$, for some time interval $[T,T+\delta]$, where the tenor $\delta>0$ is typically one day (1D), one week (1W) or several months (typically 1M, 2M, 3M, 6M or 12M). 
While before the last financial crisis rates associated to different tenors were simply related by no-arbitrage arguments, nowadays, for every tenor $\delta \in \{\delta_1, \ldots, \delta_m\}$, a specific yield curve is constructed from market instruments that depend on Libor rates corresponding to the specific tenor $\delta$. 

The rate for overnight borrowing, denoted by $L_T(T,T+1/360)$, is the Federal Funds rate in the US market and the Eonia (euro overnight index average) rate in the Euro area. 
Overnight rates represent the underlying of overnight indexed swaps (OIS) and OIS rates are the market quotes for these swaps (see Section~\ref{sec:noopt}). OIS rates play an important role, being commonly assumed to be the best proxy for risk-free rates, and are also used as collateral rates in collateralized transactions, thus leading to OIS discounting (see Appendix~\ref{sec:FRArates}). 
By relying on bootstrapping techniques (see e.g.~\cite{AB:13}), the following curves can be obtained from OIS rates:
\begin{itemize}
\item (risk-free) OIS zero coupon bond prices $T\mapsto B(t,T)$;
\item instantaneous (risk-free) OIS forward rates $T \mapsto f_t(T)=-\partial_T\log B(t,T)$;
\item simply compounded (risk-free) OIS forward rates 
\[
T\mapsto L^{D}_t(T,T+\delta) := \frac{1}{\delta}\left(\frac{B(t,T)}{B(t,T+\delta)}-1\right).
\]
\end{itemize}
In particular, note that $L^{D}_t(T,T+\delta)$ corresponds to the pre-crisis risk-free forward Libor rate at time $t$ for the interval $[T, T+\delta]$ (the superscript $D$ stands for discounting).

While OIS rates provide a complete picture of the \emph{risk-free} (discounting) yield curve, the underlying quantities of typical fixed income products, such as forward rate agreements (FRAs), swaps, caps/floors and swaptions, are Libor rates $L_T(T,T+\delta)$ for some tenor $\delta>1/360$. Since these rates are affected by the credit and liquidity risk of the panel of contributing banks (interbank risk), we shall sometimes refer to Libor rates as \emph{risky} rates. 

Among all financial contracts written on Libor rates, FRAs can be rightfully considered -- due to the simplicity of their payoff -- as the most fundamental instruments and are also liquidly traded on the derivatives' market, especially for short maturities. 
Moreover, typical linear interest rate derivatives, like swaps or basis swaps  can be represented as portfolios of FRAs (see Section~\ref{sec:noopt}).
The FRA rate at time $t$ for the interval $[T,T+\delta]$, denoted by $L_t(T,T+\delta)$, is the rate fixed at time $t$ such that the fair value of a FRA contract is null.
As shown in Appendix \ref{sec:FRArates}, the no-arbitrage value of the FRA rate $L_t(T,T+\delta)$ in line with current market practice is given by the following expression:
\begin{align}	\label{eq:defLibor}
L_t(T,T+\delta)=\mathbb{E}^{\mathbb{Q}^{T+\delta}}\left[L_T(T,T+\delta)\, \big|\, \mathcal{F}_t\right],
\end{align}
where $\mathbb{Q}^{T+\delta}$ denotes a $(T+\delta)$-forward measure with the OIS bond $B(\cdot,T+\delta)$ as num\'eraire. 
In particular, $\bigl(L_t(T,T+\delta)\bigr)_{t\in[0,T]}$ is a $\mathbb{Q}^{T+\delta}$-martingale, for all $T\geq0$, which will be the crucial property that has to be satisfied when setting up a multiple yield curve model.
Formula \eqref{eq:defLibor} has been first introduced as a definition of the FRA rate in \cite{mer10b}.

The spreads mentioned at the very beginning of the present paper arise from the fact that market FRA rates are typically higher than simply compounded OIS forward rates, i.e., $L_t(T,T+\delta) > L^D_t(T,T+\delta)$.
This is related to the fact that the Libor panel is periodically updated to include only creditworthy banks. Hence, Libor rates incorporate the risk that the average credit quality of an initial set of banks deteriorates over the term of the loan, while OIS rates  reflect the average credit quality of a newly refreshed Libor panel (see, e.g.,~\cite{fitr12}).
Therefore, since the year 2007, we observe positive spreads between FRA and OIS forward rates, as illustrated in Figures~\ref{fig:1} and~\ref{fig:2}. 
In particular, observe that spreads are generally positive and increasing with respect to the tenor length $\delta$.

\begin{figure}
\begin{minipage}[hbt]{7cm}
	\centering
	\includegraphics[width=6cm]{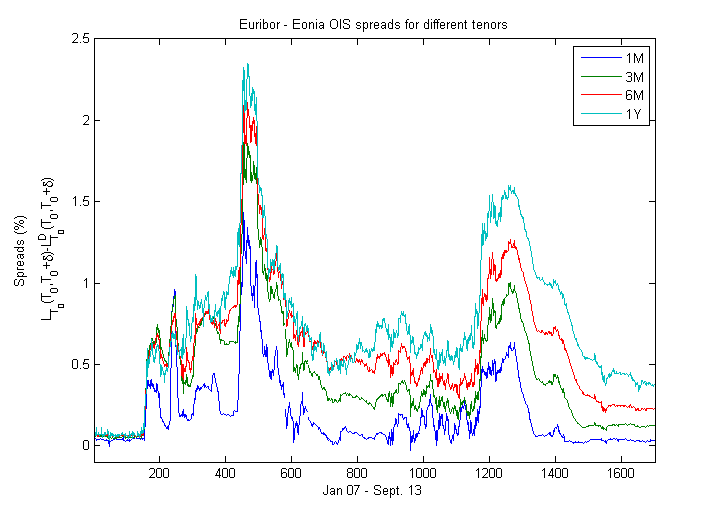}
	\caption{\small{Additive Eonia--Euribor spreads from Jan.~2007 to Sept.~2013 for $\delta=1/12, 3/12, 6/12,1$}}
	\label{fig:1}
\end{minipage}
\hfill
\begin{minipage}[hbt]{7cm}
	\centering
	\includegraphics[width=7cm]{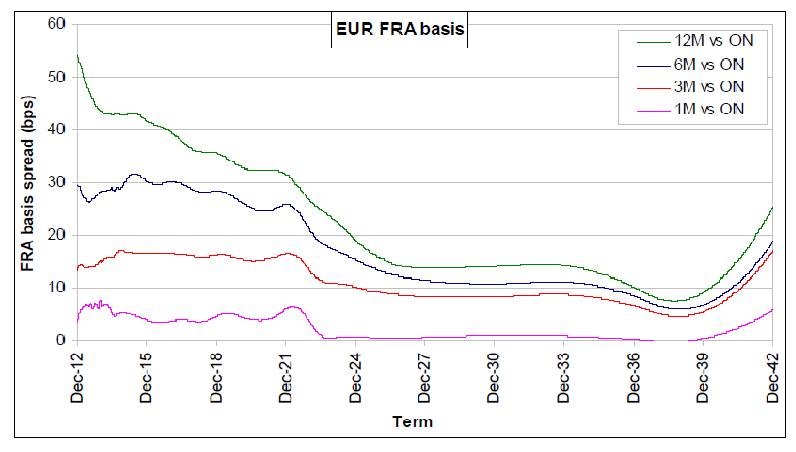}
	\caption{\small{Term structure of additive spreads between FRA rates and OIS forward rates, at Dec.~11, 2012 for $\delta=1/12, 3/12, 6/12,1$.}}
	\label{fig:2}
\end{minipage}
\end{figure}

\subsection{Problem formulation and modeling approach}	\label{subsec:model}

Motivated by the above discussion, we consider OIS zero coupon bonds and FRA contracts, for different tenors $\{\delta_1,\ldots,\delta_m\}$, as the market fundamentals. Note that, in the post-crisis interest rate market, FRA contracts must be added to the market composed of all risk-free zero coupon bonds, because they cannot be perfectly replicated by the latter any longer. Our main goal consists in solving in a general way the following problem.

\begin{problem}		\label{main_pb}
Given today's prices of OIS zero coupon bonds $B(0,T)$ and FRA rates $L_0(T,T+\delta)$, for different tenors $\delta\in\{\delta_1, \ldots, \delta_m\}$ and for all maturities $T\geq0$, model their stochastic evolution so that the market consisting of all OIS zero coupon bonds and all FRA contracts is free of arbitrage.
\end{problem}

Let us remark that throughout the paper we identify ``no-arbitrage'' with the existence of an equivalent measure under which the
OIS zero coupon bonds and FRA contracts denominated in units of the OIS bank account are martingales. 
Strictly speaking this is only a sufficient condition which guarantees ``no asymptotic free lunch with vanishing risk'' recently introduced in~\cite{CKT14}. This notion is an extension of the well-known ``no free lunch with vanishing risk'' condition by Delbaen and Schachermayer~\cite{DS94} to markets with uncountably many assets, as considered in our setting.

Apart from the presence of FRA contracts, due to the fact that Libor rates are no longer risk-free, Problem~\ref{main_pb} is the question dealt within the classical Heath-Jarrow-Morton (HJM) framework~\cite{HJM:92}, which describes the arbitrage-free stochastic evolution of the term structure of risk-free zero coupon bond prices. 
Our approach consists thus in extending the classical HJM framework in order to include FRA contracts for a finite collection of tenors $\{\delta_1,\ldots,\delta_m\}$ and for all maturities. The crucial question is how to preclude arbitrage in this setting and how to translate such a fundamental requirement into a transparent condition on the model's ingredients. 

From a modeling perspective, a first possibility would be to directly specify some dynamics for $\bigl(L_t(T,T+\delta)\bigr)_{t\in[0,T]}$, for all $\delta\in\{\delta_1,\ldots,\delta_m\}$ and $T\geq0$. However, it is easier to model spreads directly in order to capture their positivity and monotonicity with respect to the tenor, in line with the empirical findings reported above. We consider the following \emph{multiplicative} spreads
\begin{align}\label{eq:multspread}
S^{\delta}(t,T):=\frac{1+\delta L_t(T,T+\delta)}{1+\delta L^D_t(T,T+\delta)},
\end{align}
for $\delta\in\{\delta_1,\ldots,\delta_m\}$, corresponding to multiplicative spreads between (normalized) FRA rates and (normalized) simply compounded OIS forward rates. Note also that the initial curve of multiplicative spreads can be directly obtained from the OIS and FRA rates observed on the market. As an example, Figure~\ref{fig:3} displays the curve of $T \mapsto S^{\delta}(T_0,T)$ obtained from market data at $T_0=\textrm{Aug. }8, 2013$.

Let us now explain the reasoning behind this modeling choice, considering the case of a single tenor $\delta$ for simplicity of presentation. 
As a preliminary step, we illustrate a foreign exchange analogy, inspired by~\cite{JT:95} and \cite{bia10} (see Appendix~\ref{app:FX} for more details). To the risky Libor rates $L_T(T,T+\delta)$ one can associate \emph{artificial} risky bond prices $B^{\delta}(t,T)$ such that $1+\delta L_T(T,T+\delta)=1/B^{\delta}(T,T+\delta)$, for all $T\geq0$, in analogy to the classical single curve risk-free setting. 
Following for instance the discussion in~\cite{mor09}, we can think of such risky bonds as issued by a bank representative of the Libor panel\footnote{Artificial risky bonds have been introduced in a number of recent papers, see e.g.~\cite{cre12,GM:14,hen10}. We want to emphasize that artificial risky bond prices are only introduced here as an explanatory tool and shall not be considered in the following sections of the paper.}. 
If we interpret risk-free bonds as \emph{domestic} bonds and artificial risky bonds as \emph{foreign} bonds, the quantity $B^{\delta}(t,T)$ represents the price (in units of the foreign currency) of a foreign risky zero-coupon bond.
Due to equation \eqref{eq:multspread}, the (spot) multiplicative spread (for $t=T$) satisfies $S^{\delta}(T,T)=B(T,T+\delta)/B^{\delta}(T,T+\delta)$.  
According to this foreign exchange analogy as explained in detail in Appendix~\ref{app:FX}, the quantity $S^{\delta}(T,T)$ can be interpreted as the forward exchange premium between the domestic and the foreign economy over the period $[T,T+\delta]$,   measuring the change in 
the riskiness of foreign bonds with respect to domestic bonds as anticipated by the market at time $T$. 
In the present context, $S^{\delta}(T,T)$ thus represents a market valuation (at time $T$) of the credit and liquidity quality of the Libor panel (corresponding to the foreign economy) over the period $[T,T+\delta]$.
In that sense, multiplicative spreads are a natural quantity to model in a multiple yield curve framework. 

Our approach consists in formulating a general HJM framework for the term structures of OIS bond prices $B(t,T)$ and of multiplicative spreads $S^{\delta}(t,T)$. While in the case of OIS bonds the situation is analogous to the classical HJM setting, the modeling of multiplicative spreads is much less standard. To this effect, we propose an approach inspired by the HJM philosophy, as put forward in ~\cite[Section 2.1]{KK13} (compare also with~\cite{Carmona}). 

In HJM-type models there typically exists a \emph{canonical underlying asset} or a \emph{reference process} which is the underlying of the assets of interest. In our context, the assets of interest are OIS zero coupon bonds and FRA contracts.
In the case of OIS bonds, the canonical underlying asset is the (risk-free) OIS  bank account. Concerning FRA contracts, the choice is less obvious. Inspired by the foreign exchange analogy discussed above, we consider as reference process the quantity
\[
Q^{\delta}_T := S^{\delta}(T,T) = \frac{B(T,T+\delta)}{B^{\delta}(T,T+\delta)}
\quad\text{for all } T \geq 0.
\]

In order to obtain a convenient parametrization (``codebook'') of the term structures, the next step in the formulation of an HJM-type model consists in specifying \emph{simple models} for the evolution of the canonical underlying assets/reference processes.

In the case of OIS bonds, this is done by supposing that the OIS bank account, denoted by $(B_t)_{t\geq 0}$, is simply given by  $B_t=\exp\bigl(\int_0^t r_s\,ds\bigr)$, where $(r_t)_{t\geq 0}$ is a deterministic short rate. This yields the relation $r_T:=-\partial_T\log\bigl(B(t,T)\bigr)$. However, market data do not follow such a simple model and, hence, $-\partial_T\log\bigl(B(t,T)\bigr)$ yields a parameter manifold which changes randomly over time. This leads to instantaneous forward rates $f_t(T):=-\partial_T\log\bigl(B(t,T)\bigr)$, for which a stochastic evolution has to be specified. Absence of arbitrage is implied by the requirement that discounted bond prices are martingales, which then yields the well-known HJM \emph{drift condition}. The dynamics of the reference process, i.e., of the short rate $(r_t)_{t\geq0}$, are determined via the \emph{consistency condition}, that is, $r_t=f_t(t)$.

In the case of FRA contracts, we keep the simple model for the OIS bonds, assuming a deterministic short rate, and suppose additionally -- similarly to~\cite{KK13} -- the following simple model for $Q^{\delta}_T$ 
\begin{align}\label{eq:exchangerate}
Q^{\delta}_T=\exp(Z_T), 
\quad\text{for all } T \geq 0,
\end{align}
where $(Z_t)_{t\geq0}$ is a one-dimensional time-inhomogeneous L\'evy process under a given pricing measure $\mathbb{Q}$ (and thus under all forward measures due to the deterministic short rate). Its L\'evy exponent is denoted by $\psi(t,u)$, for $(t,u)\in \mathbb{R}_+ \times \mathbb{R}$. 
In view of equations~\eqref{eq:defLibor}-\eqref{eq:multspread} and recalling the relation $1+\delta L_T(T,T+\delta)=1/B^{\delta}(T,T+\delta)$, this leads to the following representation of $S^{\delta}(t,T)$:\footnote{Due to the deterministic short rate, it is not necessary to distinguish the expectations with respect to different measures, but we explicitly indicate them for consistency of the exposition with the general setting of the following sections.}
\begin{align}
S^{\delta}(t,T)
&= \frac{B(t,T+\delta)}{B(t,T)}\mathbb{E}^{\mathbb{Q}^{T+\delta}}\bigl[1+\delta L_T(T,T+\delta)\,\bigr|\,\mathcal{F}_t\bigr]
= \frac{B(t,T+\delta)}{B(t,T)}\mathbb{E}^{\mathbb{Q}^{T+\delta}}\left[\frac{1}{B^{\delta}(T,T+\delta)}\,\Bigr|\,\mathcal{F}_t\right]	\nonumber\\
&= \mathbb{E}^{\mathbb{Q}^T}\left[\frac{B(T,T+\delta)}{B^{\delta}(T,T+\delta)}\,\Bigr|\,\mathcal{F}_t\right]
= \mathbb{E}^{\mathbb{Q}^T}\bigl[Q^{\delta}_T\,|\,\mathcal{F}_t\bigr]
= \mathbb{E}^{\mathbb{Q}^T}\left[e^{Z_T}\,|\,\mathcal{F}_t\right]	\nonumber\\
&= \exp\Bigl(Z_t+\int_t^T\! \psi(s,1)\,ds\Bigr).	\label{eq:formspread}
\end{align}
In particular, this implies the following relation:
\begin{equation}	\label{eq:spreadrate}
\partial_T\log\bigl(S^{\delta}(t,T)\bigr) = \psi(T,1).
\end{equation}
As in the case of the OIS term structure, since market data do not follow such a simple model
and $\bigl(S^{\delta}(t,T)\bigr)_{t\in[0,T]}$ evolves randomly over time, 
we need to put $\psi(T,1)$ ``in motion''. To this effect, we define an instantaneous forward spread rate via the left-hand side of \eqref{eq:spreadrate}, i.e., $\eta^{\delta}_t(T):=\partial_T\log\bigl(S^{\delta}(t,T)\bigr)$, and specify general stochastic dynamics for $\eta^{\delta}_t(T)$. A typical shape of the curve $T \mapsto \eta^{\delta}_{T_0}(T)$ is shown in
Figure~\ref{fig:4}, obtained from market data at $T_0=\textrm{Aug. }8, 2013$.
From the defining property of FRA rates (see equation~\eqref{eq:defLibor}), which is equivalent to the $\mathbb{Q}^T$-martingale property of $\bigl(S^{\delta}(t,T)\bigr)_{t\in[0,T]}$, for all $T\geq0$ (see Lemma~\ref{lem:spreadQT}), an HJM \emph{drift condition} can be deduced, which then ensures absence of arbitrage. Moreover, the dynamics of the reference process $(Q^{\delta}_t)_{t\geq0}$, or, equivalently, those of $(Z_t)_{t\geq0}$ (assumed to be a general It\^o-semimartingale in the following), have to satisfy a suitable \emph{consistency condition}, similar to the requirement $f_t(t)=r_t$ in the case of the OIS term structure.

\begin{figure}
\begin{minipage}[hbt]{7cm}
	\centering
	\includegraphics[width=7cm]{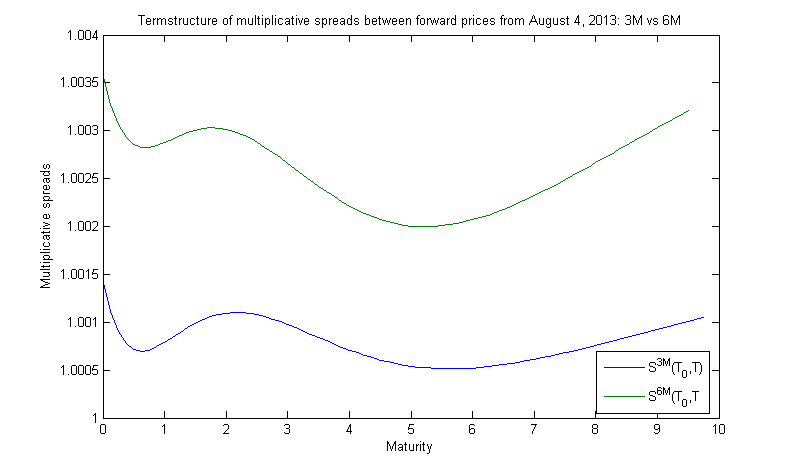}
\caption{\small{Multiplicative spreads $S^{\delta}(T_0,T)$  at Aug.~8, 2013 for $\delta=3/12, 6/12$.}}
	\label{fig:3}
\end{minipage}
\hfill
\begin{minipage}[hbt]{7cm}
	\centering
	\includegraphics[width=7cm]{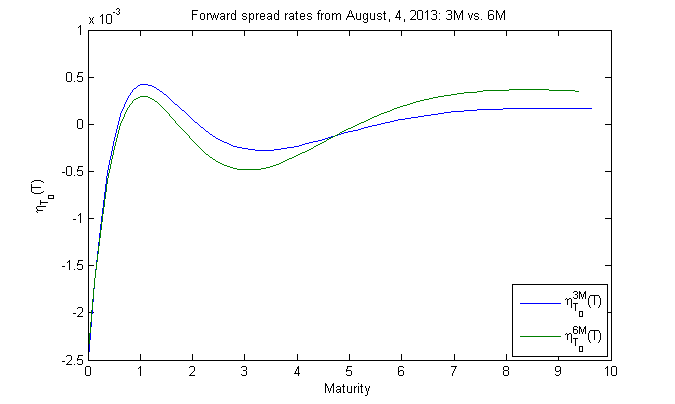}
	\caption{\small{Forward spread rates $\eta_{T_0}(T)$  at Aug.~8, 2013  for $\delta=3/12, 6/12$.}}
	\label{fig:4}
\end{minipage}
\end{figure}

\subsection{Main features of a general HJM-type framework}

Summing up, let us highlight the main features and novel contributions of the proposed approach:

\begin{itemize}
\item 
The term structure associated to Libor rates for different tenors $\{\delta_1,\ldots,\delta_m\}$ is modeled via the multiplicative spreads $S^{\delta}(t,T)$, which are directly related to observable OIS and FRA rates. Multiplicative spreads, rather than additive ones, have a natural economic interpretation in terms of forward exchange premiums (see Appendix~\ref{app:FX}) and lead to highly tractable models, especially when the semimartingale driving $f_t(T)$ and $\eta^{\delta}_t(T)$ is an affine process (see Section~\ref{affinespecification}).
\item 
The modeling of multiplicative spreads $S^{\delta}(t,T)$ is split into two components: an instantaneous \emph{forward} spread rate $\eta^{\delta}_t(T)$ and a \emph{spot} rate $Q^{\delta}_t=S^{\delta}(t,t)$, which is directly observable from market data.  In particular, this separation allows for great modeling flexibility.
\item 
By choosing a common $\mathbb{R}^n$-valued semimartingale $Y$ for the spot spreads $S^{\delta}(t,t)$ corresponding to different tenors $\delta\in\{\delta_1\ldots,\delta_m\}$, such that $S^{\delta_i}(t,t)=\exp(u_i^{\top}Y_t)$ for $u_i \in \mathbb{R}^n$, the inherent dependence between spreads associated to different tenors (as visible from Figure \ref{fig:1}) can be captured\footnote{Note that, for every tenor $\delta\in\{\delta_1,\ldots,\delta_m\}$, the process $u_i^{\top}Y$ plays the role of the process $Z$ appearing in~\eqref{eq:exchangerate}.}. Moreover, complex correlation structures between OIS and FRA rates can be built in through a common driving process for the forward rates $f_t(T)$ and the spread rates $\eta^{\delta}_t(T)$.
\item 
The desired feature that  $S^{\delta}(t,T) \geq 1$, for all $0\leq t\leq T$, can be easily achieved in full generality. Moreover, we can easily characterize order relations between spreads associated to different tenors, i.e, when $S^{\delta_j}(t,T) \geq S^{\delta_i}(t,T) $ for $\delta_j \geq \delta_i$ and for all $0\leq t\leq T$, as is the case in typical market situations.
\item
When considering finite-dimensional factor models, we are naturally led to the class of affine processes. In this case, the model for the OIS term structure becomes a classical short rate model driven by a multidimensional affine process, which also determines the dynamics of the multiplicative spreads. In this context, one can obtain tractable valuation formulas for derivatives written on Libor rates, as shown in the companion paper \cite{CFGaffine}.
\end{itemize}

\section{The general HJM modeling framework}\label{genFram}

In this section, following the ideas introduced in the previous section, we introduce a general HJM-type framework for multiple yield curve modeling. We start in Section~\ref{sec:1.1} by defining an abstract setting where we consider general families of semimartingales. In Sections~\ref{sec:HJM} and \ref{sec:model_spread}, we then apply this to the term structure modeling of OIS zero coupon bond prices and of the multiplicative spreads defined in~\eqref{eq:multspread}, respectively.

\subsection{Abstract HJM setting}\label{sec:1.1}

Let  $(\Omega, \mathcal{F}, (\mathcal{F}_t)_{t\geq 0}, \mathbb{Q})$ be a stochastic basis endowed with a right-continuous filtration $(\mathcal{F}_t)_{t\geq0}$ supporting the processes introduced in this section. We aim at modeling a family of one-dimensional positive semimartingales 
$\{(S(t,T))_{t\in[0,T]}, T \geq 0\bigr\}$
such that $(S(t,t))_{t \geq 0}$ is also a (positive) semimartingale. Supposing differentiability of $T \mapsto \log\left(S(t,T)\right)$ (a.s.), we can represent $S(t,T)$ by 
\begin{align}\label{eq:rep}
S(t,T)=e^{Z_t+ \int_t^T \eta_t(u)du},
\end{align}
where $Z_t:=\log(S(t,t))$ and $\eta_t(T):=\partial_T\log\left(S(t,T)\right)$.
Modeling the family $\{(S(t,T))_{t\in[0,T]}, T \geq 0\}$ is thus equivalent to modeling $(Z_t)_{t \geq0}$ and $\{(\eta_t(T))_{t\in[0,T]}, T \geq 0\}$. We call $Z$ the \emph{log-spot rate} and $\eta_t(T)$ the \emph{generalized forward rate}.

The representation \eqref{eq:rep} is motivated by the HJM philosophy as discussed in Section~\ref{subsec:model}. Indeed, suppose that the spot process $S(t,t)$ corresponds to a canonical underlying asset and that $S(t,T)=\mathbb{E}[S(T,T)|\mathcal{F}_t]$, for all $0\leq t\leq T$.
If $S(t,t)$ is modeled as an exponential time-inhomogeneous L\'evy process $\exp(Z_t)$, we obtain expression~\eqref{eq:formspread} for $S(t,T)$ (under the measure $\QQ$). Putting the L\'evy exponent ``in motion'' naturally leads to \eqref{eq:rep} with a general semimartingale $Z$.

We define \emph{HJM-type models} as follows (compare with~\cite[Definition 3.1]{KK13}).

\begin{definition}\label{def:HJMtype}
A quintuple $(Z, \eta_0, \alpha, \sigma, X)$ is called \emph{HJM-type model} for the family of positive semimartingales 
$\{(S(t,T))_{t\in[0,T]}, T \geq 0\}$  if
\begin{enumerate}
\item $(X,Z)$ is an $\RR^{d+1}$-valued It\^o-semimartingale, i.e., its characteristics are absolutely continuous with respect to the Lebesgue measure (see e.g.~\cite[Definition 2.1.1]{JP12});
\item $\eta_0$: $\mathbb{R}_+ \to \mathbb{R}$ is measurable and $\int_0^T |\eta_0(u)|du < \infty$ $\QQ$-a.s. for all $T \in \mathbb{R}_+$;
\item $(\omega,t, T)\mapsto \alpha_t(T)(\omega)$ and $(\omega,t, T)\mapsto \sigma_t(T)(\omega)$
are $\mathcal{P} \otimes \mathcal{B}(\mathbb{R}_+)$-measurable $\mathbb{R}$- and $\mathbb{R}^d$-valued processes, respectively, where $\mathcal{P}$ denotes the predictable $\sigma$-algebra, and satisfy 
\begin{itemize}
\item  $\int_0^t \int_0^T |\alpha_s(u)| du ds < \infty$ $\QQ$-a.s. for all $t,T \in \mathbb{R}_+$,
\item  $\int_0^T \| \sigma_t(u) \|^2 du < \infty$ $\QQ$-a.s. for any $t,T \in \mathbb{R}_+$,
\item  $\bigl((\int_0^T |\sigma_{t,j}(u)|^2 du)^{\frac{1}{2}}\bigr)_{t\geq0} \in L(X^j)$ for all $ T\in \mathbb{R}_+$ and $j \in \{1,\ldots, d\}$, where $L(X^j)$ denotes the set of processes which are integrable with respect to $X^j$;
\end{itemize}
\item for every $T \in \mathbb{R}_+$, the generalized forward rate $(\eta_t(T))_{t\in [0,T]}$ is given by, for all $t\leq T$,
\begin{align}\label{eq:eta}
\eta_t(T)=\eta_0(T)+\int_0^t \alpha_s(T)ds+\int_0^t \sigma_s(T) dX_s;
\end{align}
\item $\{ (S(t,T))_{t \in [0,T]}, T \geq 0 \}$ satisfies, for all $t\leq T$ and $T\in\mathbb{R}_+$,
\begin{align}\label{eq:HJMmodel}
S(t,T)=e^{Z_t+\int_t^T \eta_t(u)du}
\end{align}
and, in particular, $S(t,t)=e^{Z_t}$ for all $t\geq0$.
\end{enumerate}
\end{definition}

Typically, $(S(t,T))_{t \in [0,T]}$ corresponds to a discounted asset price process with maturity $T$ and is thus -- under the assumption of ``no free lunch with vanishing risk''-- a (local) martingale under some equivalent measure. Supposing that $\mathbb{Q}$ already represents a (local) martingale measure, the (local) martingale property of $(S(t,T))_{t \in [0,T]}$ can be characterized by Theorem~\ref{th:mart} below (see also Remark~\ref{rem:loc_mart}). As a preliminary, we recall the notion of \emph{local exponent} (compare~\cite[Definition A.6]{KK13}) (or, equivalently, \emph{derivative of the Laplace cumulant process}, see~\cite[Definitions 2.22 and 2.23]{KS:02})\footnote{The definition of local exponent in~\cite[Definition A.6]{KK13} is slightly different since it is defined in terms of the exponential compensator of $(\int_0^t \im \beta_s dX_s)_{t\geq0}$ which is due to the fact that complex valued processes are considered in that paper. Note also that, since It\^o-semimartingales are quasi-left-continuous, the derivatives of the modified Laplace cumulant process and of the ordinary Laplace cumulant process coincide (see e.g.~\cite[Remarks on page 408]{KS:02}).}.

\begin{definition}\label{def:expcomp}
Let $X$ be an $\mathbb{R}^d$-valued It\^o-semimartingale and $\beta=(\beta_t)_{t\geq0}$ an $\mathbb{R}^d$-valued predictable X-integrable process (i.e., $\beta\in L(X)$).
A predictable real-valued process  $\bigl(\Psi_t^X(\beta_t)\bigr)_{t\geq0}$ is called \emph{local exponent} (or derivative of the Laplace cumulant process) of $X$ at $\beta$ if 
$
\bigl(\exp \bigl(\int_0^t \beta_s dX_s - \int_0^t \Psi_s^X(\beta_s) ds \bigr)\bigr)_{t\geq0}
$
is a local martingale. We denote by $\mathcal{U}^X$ the set of processes $\beta$ such that $\Psi^X(\beta)$ exists.
\end{definition}

In other words, $\bigl(\int_0^t \Psi_s^X(\beta_s) ds\bigr)_{t\geq0}$ is the \emph{exponential compensator} (see~\cite[Definition 2.14]{KS:02}) of $\bigl(\int_0^t \beta_s dX_s\bigr)_{t\geq0}$.
The following proposition asserts that the local exponent (when it exists) is of L\'evy-Khintchine form, where the L\'evy triplet is replaced by the differential characteristics of the It\^o-semimartingale.

\begin{proposition}\label{prop:expcomp}
Let $X$ be an $\mathbb{R}^d$-valued It\^o-semimartingale with differential characteristics $(b, c,K)$ with respect to some truncation function $\chi$. Let $\beta\in L(X)$. Then the following are equivalent:
\begin{enumerate}
\item $\beta \in \mathcal{U}^X$;
\item $\bigl(\int_0^{t} \beta_s dX_s\bigr)_{t\geq0}$ is an exponentially special semimartingale, that is $\bigl(e^{\int_0^{t} \beta_s dX_s}\bigr)_{t\geq0}$ is a special semimartingale;
\item $\int_0^t \int_{\beta_s^{\top} \xi >1 } e^{\beta_s^{\top}\xi}K_s(d\xi)ds <\infty$ $\QQ$-a.s for all $t\geq0$.
\end{enumerate}
In this case, outside some $d\mathbb{Q} \otimes dt$ nullset, it holds that
\begin{align}\label{eq:LevyKinthchine}
\Psi_t^X(\beta_t)=\beta_t^{\top} b_t+\frac{1}{2}\beta^{\top}_t c_t \beta_t+\int \bigl(e^{\beta^{\top}_t \xi}-1-\beta_t^{\top} \chi(\xi)\bigr) K_t(d\xi).
\end{align}
\end{proposition}

\begin{proof}
For the proof of the equivalence of (i)-(ii)-(iii) see~\cite[Lemma 2.13]{KS:02}. Representation~\eqref{eq:LevyKinthchine} follows from~\cite[Theorem 2.18, statements 1 and 6, and Theorem 2.19]{KS:02}.
\end{proof}

Using the notion of the local exponent and defining an $\mathbb{R}^d$-valued process $(\Sigma_t(T))_{t \in [0,T]}$ via 
\[
\Sigma_t(T):=\int_t^T \sigma_t(u)du,
\]
for all $t\leq T$ and $T\geq0$, we are now in a position to state the following theorem, which characterizes the martingale property of the family of semimartingales $\{(S(t,T))_{t\in[0,T]}, T \geq 0\}$.
We denote by $\Psi^{Z,X}$ the local exponent of the $\mathbb{R}^{1+d}$-valued semimartingale $(Z,X)$.

\begin{theorem}\label{th:mart}
For an HJM-type model as of Definition~\ref{def:HJMtype} the following conditions are equivalent:
\begin{enumerate}
\item the process $(S(t,T))_{t \in [0,T]}$ is a martingale, for every $T \geq 0$.
\item for every $T \geq 0$, it holds that
\[
\mathbb{E}\left[e^{Z_T}| \mathcal{F}_t\right]=e^{Z_t+ \int_t^T \eta_t(u)du},
\qquad\text{ for all }t\in[0,T],
\]
which is called \emph{conditional expectation hypothesis}.
\item  for every $T \geq 0$, $\bigl(1, \Sigma^{\top}_{\cdot}(T)\bigr)^{\top} \in \mathcal{U}^{Z,X}$ and the following conditions are satisfied:
\begin{itemize}
\item The process
\begin{align}\label{eq:martprop}
\left(\exp \left(Z_t+\int_0^t \Sigma_s(T) dX_s - \int_0^t
\Psi_s^{Z,X}\left(\bigl(1, \Sigma^{\top} _s(T)\bigr)^{\top}\right)
ds \right)\right)_{t\in[0,T]}
\end{align}
is a martingale, for every $T \geq 0$.
\item the \emph{consistency condition} 
\begin{align}\label{eq:consistency}
\Psi_t^{Z}(1)=\eta_{t-}(t), 
\qquad\text{ for all } t>0,
\end{align}
holds.
\item the \emph{HJM drift condition}
\begin{align}\label{eq:drift}
\int_t^T \alpha_t(u) du=\Psi_t^{Z}(1)-\Psi_t^{Z,X}\left(\bigl(1, \Sigma^{\top}_t(T)\bigr)^{\top}\right)
\end{align}
holds for every $t \in [0,T]$ and $T \geq 0$.
\end{itemize}
\end{enumerate}
Moreover, if any (and, hence, all) of conditions (i)-(ii)-(iii) is satisfied, it holds that
\begin{equation}\label{eq:eq}
\begin{split}
S(t,T)&=\mathbb{E}\left[S(T,T)|\mathcal{F}_t\right]=\mathbb{E}\left[e^{Z_T}|\mathcal{F}_t\right]\\
&=\exp\left(\int_0^T \eta_0(u)du + Z_t+ \int_0^t \Sigma_s(T) dX_s-\int_0^t \Psi_s^{Z,X}\left(\bigl(1, \Sigma^{\top}_s(T)\bigr)^{\top}\right)ds\right).
\end{split}
\end{equation}
for all $t\leq T$ and $T\geq0$.
\end{theorem}

\begin{remark}	\label{rem:loc_mart}
The above theorem admits a local martingale version, in the sense that $(S(t,T))_{t\in[0,T]}$ is a local martingale if and only if $\bigl(1,\Sigma_{\cdot}^{\top}(T)\bigr)^{\top}\in\mathcal{U}^{Z,X}$ and the consistency and HJM drift conditions \eqref{eq:consistency}-\eqref{eq:drift} hold. 
\end{remark}

\begin{remark}\label{rem:martcond}
General sufficient conditions for~\eqref{eq:martprop} being a true martingale 
have been established by Kallsen and Shiryaev in~\cite[Corollary 3.10]{KS:02} (in that context, see also the recent paper~\cite{LR14}). 
For instance, condition $I(0,1)$ in their formulation (see~\cite[Definition 3.1]{KS:02} and compare also with~\cite[Proposition 3.3]{GG14}) reads in our case as
\[
\int_0^t \int_{|\left(1,\Sigma^{\top}_s(T)\right)\xi| >1} e^{\left(1,\Sigma^{\top}_s(T)\right)^{\top}\xi}\bigl|\bigl(1,\Sigma^{\top}_s(T)\bigr)^{\top}\xi\bigr| 
 \,K_s^{Z,X}(d\xi)ds<\infty \text{ $\mathbb{Q}$-a.s. for all }t\geq0,
\]
together with
\begin{align*}
&\sup_{t\leq T} \mathbb{E}\Bigg[\exp\left(\frac{1}{2}\int_0^t\bigl(1,\Sigma^{\top}_s(T)\bigr)^{\top} c_s^{Z,X}\bigl(1,\Sigma^{\top}_s(T)\bigr)ds\right)\\
&\qquad\quad\times \exp\left(\int_0^t\int \left(e^{\bigl(1,\Sigma^{\top}_s(T)\bigr)^{\top}\xi}\left(\bigl(1,\Sigma^{\top}_s(T)\bigr)^{\top}\xi-1\right)+1\right)K_s^{Z,X}(d\xi)ds\right)\Bigg]< \infty,
\end{align*}
where $c^{Z,X}$ and $K^{Z,X}$ denote the second and third characteristic of the semimartingale $(Z,X)$.
In particular, observe that the first of the above two conditions implies that $\bigl(1,\Sigma_{\cdot}^{\top}(T)\bigr)^{\top}\in\mathcal{U}^{Z,X}$.
\end{remark}

\begin{proof}[Proof of Theorem~\ref{th:mart}]
In the sequel, let $T >0$ be fixed.\\
(i) $\Rightarrow$ (ii): From~\eqref{eq:HJMmodel} and the martingale property of $(S(t,T))_{t\in [0,T]}$, it follows that
\[
e^{Z_t+\int_t^T \eta_t(u) du}=S(t,T)=\mathbb{E}\left[S(T,T)| \mathcal{F}_t\right]=\mathbb{E}\left[e^{Z_T}| \mathcal{F}_t\right],
\qquad\text{ for all }t\in[0,T],
\]
whence (ii).\\
(ii)  $\Rightarrow$ (iii): Let us define $R_t:=Z_t+\int_t^T \eta_t(u)du$, for all $t\leq T$. Then the martingale property of $S(\cdot,T)=\exp(R)$ and Definition~\ref{def:expcomp} yields $1 \in \mathcal{U}^R$ and
$\Psi_t^R(1)=0$. By applying the classical and the stochastic Fubini theorem~\cite[Theorem IV.65]{protter}, which is justified due to the integrability conditions on $\alpha$ and $\sigma$ in Definition~\ref{def:HJMtype}, we can write
\begin{align*}
\int_t^T \eta_t(u)du&=\int_0^T \eta_0(u)du+\int_0^t\int_s^T  \alpha_s(u) du ds +\int_0^t \Sigma_s(T) dX_s\\
&\quad - \int_0^t \underbrace{\left(\eta_0(u)+ \int_0^u \alpha_s(u) ds+\int_0^u  \sigma_s(u) dX_s\right)}_{= \eta_{u}(u)} du.
\end{align*}
We thus obtain (e.g., by applying~\cite[Lemma A.20]{KK13})
\begin{align}\label{eq:driftcons}
0=\psi_t^R(1)=\psi_t^{Z,X}\Bigl(\bigl(1, \Sigma^{\top}_t(T)\bigr)^{\top}\Bigr)+\int_t^T \alpha_t(u) du-\eta_{t-}(t).
\end{align}
Setting $T=t$ and noting that $\Sigma_t(t)=0$ yields~\eqref{eq:consistency}, namely $
\eta_{t-}(t)=\psi_t^{Z}(1)$, for all $t>0$, which together with~\eqref{eq:driftcons} then implies~\eqref{eq:drift}. By the drift and consistency conditions, 
$S(\cdot,T)$ is then of the form~\eqref{eq:eq} and since $(S(t,T))_{t \in [0,T]}$ is a martingale, the martingale property holds for~\eqref{eq:martprop} as well.\\
(iii) $\Rightarrow$ (i): The martingale property of~\eqref{eq:martprop} implies the martingale property of $(S(t,T))_{t \in [0,T]}$, since -- due to the drift and consistency condition -- it is again of the form~\eqref{eq:eq}. 
This clearly proves also the last statement of the theorem.
\end{proof}

\subsection{Modeling the term structure of OIS zero coupon bond prices}\label{sec:HJM}

In this section, we show that the classical HJM approach for risk-free bond prices, which we use for modeling OIS bonds, can be formulated in terms of the above general framework (compare also with~\cite{T:pres} by J.~Teichmann). We start by defining (OIS) bond price models via a model for the instantaneous (OIS) forward rates and the assumption that an (OIS) bank account $B$ exists, given by 
\[
B_t=e^{\int_0^t r_s ds},
\]
where $r$ denotes the (OIS) short rate process.

\begin{definition} 	\label{def:bond_price_model}
A \emph{bond price model} is a quintuple $(B, f_0, \widetilde{\alpha}, \widetilde{\sigma}, X)$,
where
\begin{enumerate}
\item the bank account $B$ satisfies $B_t=e^{\int_0^t r_s ds}$, for all $t\geq0$, with short rate $(r_t)_{t\geq 0}$;
\item $X$ is an $\RR^d$-valued It\^o-semimartingale;
\item $f_0$: $\mathbb{R}_+ \to \mathbb{R}$ is measurable and $\int_0^T |f_0(t)|dt < \infty$ $\QQ$-a.s. for all $T\geq0$;
\item $(\omega,t, T)\mapsto \widetilde{\alpha}_t(T)(\omega)$ and $(\omega,t, T)\mapsto \widetilde{\sigma}_t(T)(\omega)$
are $\mathcal{P} \otimes \mathcal{B}(\mathbb{R}_+)$ measurable $\mathbb{R}$- and $\mathbb{R}^d$-valued processes and satisfy 
the integrability conditions of Definition~\ref{def:HJMtype}-(iii);
\item for every $T \in \mathbb{R}_+$, the forward rate $(f_t(T))_{t\in [0,T]}$ is given by
\begin{equation}	\label{eq:fwd_classical}
f_t(T)=f_0(T)+\int_0^t \widetilde{\alpha}_s(T)ds+\int_0^t \widetilde{\sigma}_s(T) dX_s;
\end{equation}
\item the bond prices $\{ (B(t,T))_{t \in [0,T]}, T \geq 0 \}$ satisfy $B(t,T)=e^{-\int_t^T f_t(s)ds}$, for all $t\leq T$ and $T\geq0$. In particular, $B(t,t)=1$ for all $t\geq0$.  
\end{enumerate}
\end{definition}

The following definition is motivated by the fact that, as usual, we take
the (OIS) bank account as num\'eraire and assume that $\mathbb{Q}$ corresponds to a risk neutral measure (see however Remark~\ref{rem:terminal_bond}). In the following, if the measure is not explicitly indicated (e.g., in expectations), then it is meant to be $\mathbb{Q}$.

\begin{definition}	\label{def:bond_RN}
A bond price model is said to be \emph{risk neutral} if the discounted bond prices 
\[
\bigl\{ \bigl(B(t,T)/B_t\bigr)_{t \in [0,T]}, \, T \geq 0 \bigr\}
\] 
are martingales.
\end{definition}

The following proposition shows that a bond price model can be identified with an HJM-type model. For its formulation, let us introduce the process $\widetilde{\Sigma}_{\cdot}(T)$ defined via $\widetilde{\Sigma}_t(T):=\int_t^T \widetilde{\sigma}_t(u)du$, for all $t\leq T$.

\begin{proposition}\label{prop:bondHJM}
A bond price model can be identified with an HJM-type model $(Z,\eta_0, \alpha, \sigma, X)$ for the family of discounted bond prices $\bigl\{ \bigl(B(t,T)/B_t\bigr)_{t \in [0,T]}, \, T \geq 0 \bigr\}$
by setting $\eta_0=-f_0$, $\alpha=-\widetilde{\alpha}$, $\sigma=-\widetilde{\sigma}$ (so that $\eta_t(t)=-f_t(T)$) and $Z_t =-\log{B_t}=-\int_0^t r_s ds$.
Moreover, the following are equivalent:
\begin{enumerate}
  \item the bond price model is risk neutral, in the sense of Definition~\ref{def:bond_RN};
\item for every $T\geq0$, the conditional expectation hypothesis holds:
\[
\mathbb{E}\left[e^{Z_T}| \mathcal{F}_t\right]=e^{Z_t+\int_t^T \eta_t(u)du},
\]
or, equivalently, $\mathbb{E}\left[B_t/B_T\, |\, \mathcal{F}_t\right]=e^{-\int_t^T f_t(u) du}$, for all $t\in[0,T]$;
 \item for every $T\geq0$, $-\widetilde{\Sigma}(T) \in \mathcal{U}^X$ and the following conditions are satisfied:
 \begin{itemize}
 \item the process
 \[
\left(\exp \left(-\int_0^t\widetilde{\Sigma}_s(T) dX_s - \int_0^t
\Psi_s^{X}\bigl( -\widetilde{\Sigma}_s(T)\bigr)
ds \right)\right)_{t\in[0,T]},
\]
is a martingale, for every $T \geq 0$.
 \item the consistency condition holds, i.e.,
\[ 
\Psi_t^{Z}(1)=-r_{t-}=-f_{t-}(t), 
\qquad \text{ for all }t>0.
\]
 \item the HJM drift condition  
 \[
 \int_t^T \widetilde{\alpha}_t(u) du= \Psi_t^{X}\bigl( -\widetilde{\Sigma}_t(T)\bigr)
 \]
 holds for every $t \in [0,T]$ and $T \geq 0$.
\end{itemize}
 \end{enumerate}
\end{proposition}

\begin{proof}
The proposition follows from Theorem~\ref{th:mart} by identifying $S(t,T)$ with $B(t,T)/B_t$, so that
$Z_t=\log(S(t,t))=-\log(B_t)$ and noting that
\[
\Psi_t^{Z,X}\left(\bigl(1, -\widetilde{\Sigma}^{\top}_t(T)\bigr)^{\top}\right)=-r_{t-}+\Psi_t^{X}\bigl( -\widetilde{\Sigma}_t(T)\bigr),
\] 
which follows for instance from~\cite[Lemma A.20]{KK13}.
\end{proof}

\begin{remark}	\label{rem:terminal_bond}
We want to point out that the assumption of the existence of a bank account is actually not necessary.
Indeed, one could also consider the (OIS) bond market for maturities $T \leq T^*$ with respect to a bond  $B(\cdot,T^*)$ as num\'eraire, where $T^*$ denotes some fixed terminal maturity (in this regard, compare the recent paper~\cite{KST13}). 
In that context, suitable HJM drift and consistency conditions can be derived by adapting the arguments presented above.
\end{remark}

\subsection{Modeling the term structure of multiplicative spreads}	\label{sec:model_spread}

In this section we introduce the modeling framework for multiple yield curves, where we extend the above considered model for the OIS bonds with an HJM-type model for the multiplicative spreads introduced in~\eqref{eq:multspread}.

\subsubsection{Modeling the term structure of multiplicative spreads}
\label{subsec:model_spread}

Let $\mathcal{D}=\{\delta_1,\ldots,\delta_m\}$ denote a family of tenors, with $0<\delta_1<\delta_2<\ldots<\delta_m$, for some $m\in\mathbb{N}$. As argued in the introduction, we aim at modeling the term structure of multiplicative spreads between
normalized FRA rates and simply compounded OIS forward rates given by
\[
S^{\delta_i}(t,T)=\frac{1+\delta_i L_t(T,T+\delta_i)}{1+\delta_i L^D_t(T,T+\delta_i)},
\]
for all $i=1,\ldots,m$. Starting with time-inhomogeneous exponential L\'evy  models for the multiplicative spot spread (or ``forward exchange premium'') process $(Q^{\delta_i}_t)_{t \geq 0}$ defined in~\eqref{eq:exchangerate} and putting the L\'evy exponent (evaluated at 1) ``in motion'', as described in Section~\ref{subsec:model}, naturally leads to HJM-type models where
\[
S^{\delta_i}(t,T)=e^{Z^i_t + \int_t^{T} \eta^i_t(u)du}.
\]
In particular, this allows to model the observed log-spot spreads $Z^i_t=\log(S^{\delta_i}(t,t))$ and the forward spread rates $\eta^i_t(T)=\partial_T\log(S^{\delta_i}(t,T))$ separately. 
This feature is important in order to capture the dependence structures between different spreads (which are easiest observed on the spot level), while guaranteeing at the same time that $S^{\delta_i}(t,T)\geq 1$ for all maturities.
If desired, this also allows to easily accommodate monotonicity for all maturities with respect to the tenors $\delta_i$ (see Corollary~\ref{cor:ordered}).

As in Section~\ref{sec:HJM}, we assume to have an OIS bank account
and we work under a risk neutral measure $\mathbb{Q}$ under which 
discounted OIS bond prices $B(t,T)/B_t$ are martingales, as required in Definition~\ref{def:bond_RN}. As a consequence of the defining property of the FRA rates (specified in ~\eqref{eq:defLibor}; see also Appendix~\ref{sec:FRArates}), we obtain the following lemma, which is crucial for absence of arbitrage in our setting.

\begin{lemma}\label{lem:spreadQT}
Assume~\eqref{eq:defLibor}. Then, for every $\delta_i \in \mathcal{D}$ and $T\geq0$, the process $(S^{\delta_i}(t,T))_{t\in [0,T]}$ is a $\mathbb{Q}^T$-martingale, where $\mathbb{Q}^T$ denotes the $T$-forward measure whose density process is given by $\frac{d\mathbb{Q}^{T}}{d\mathbb{Q}}|_{\mathcal{F}_t}=\frac{B(t,T)}{B_t B(0,T)}$, $t\in[0,T]$.
\end{lemma}
\begin{proof}
For $T\geq0$, by Bayes' formula, $(S^{\delta_i}(t,T))_{t\in [0,T]}$ is a $\mathbb{Q}^T$-martingale if and only if
\[
M^i_t:=S^{\delta_i}(t,T)\frac{B(t,T)}{B_t B(0,T)}
\]
is a $\mathbb{Q}$-martingale. By definition of $S^{\delta_i}(t,T)$, the process $M^i$ can be rewritten as
\begin{align*}
M^i_t
&=(1+\delta_i L_t(T,T+\delta_i))\frac{B(t,T+\delta_i)}{B_t B(0,T+\delta_i)}\frac{B(0,T+\delta_i)}{B(0,T)},
\end{align*}
which is -- again by Bayes' formula -- a $\mathbb{Q}$-martingale, since $(1+\delta_i L_t(T,T+\delta_i))_{t\in[0,T]}$ is a $\mathbb{Q}^{T+\delta_i}$-martingale by~\eqref{eq:defLibor} and $\frac{d\mathbb{Q}^{T+\delta_i}}{d\mathbb{Q}}|_{\mathcal{F}_t}=\frac{B(t,T+\delta_i)}{B_t B(0,T+\delta_i)}$, for all $i=1,\ldots,m$.
\end{proof}

By relying on the above lemma and referring to the arguments already discussed in Section~\ref{subsec:market}, let us now summarize the modeling requirement on
$\{(S^{\delta}(t,T))_{t\in[0,T]}, T \geq 0,  \delta \in \mathcal{D}\}$.

\begin{requirement}\label{req:modelreq}
The family of spreads $\{(S^{\delta}(t,T))_{t\in[0,T]}, T \geq 0,  \delta \in \mathcal{D}\}$ should satisfy 
\begin{enumerate}
\item $(S^{\delta_i}(t,T))_{t \in [0,T]}$ is a $\mathbb{Q}^T$-martingale, for every $T \geq 0$ and for all $i\in\{1,\ldots,m\}$;
\item $S^{\delta_i}(t,T) \geq 1$ for all $t \leq T$, $T\geq 0$ and $i\in\{1,\ldots,m\}$.
\end{enumerate}
\end{requirement}

In typical market situations, it is additionally desirable to have spreads which are ordered with respect to the different tenors $\delta_i$, that is 
\[
S^{\delta_1}(t,T) \leq \cdots \leq S^{\delta_m}(t,T),
\qquad\text{for all $t \leq T$ and $T\geq 0$.} 
\]
Due to the apparent strong interdependencies between the different spot spreads associated to different tenors $\delta_i$ (compare Figure~\ref{fig:1}), we model $Z^i$ via a common lower dimensional process $Y=(Y_t)_{t\geq0}$ taking values in  $\mathbb{R}^n$ (with $n\leq m$) such that
\[
Z^i_t:=u_{i}^{\top}Y_t,
\] 
where $u_{1}, \ldots, u_{m}$ are given vectors in $\mathbb{R}^n$. The dimension of $Y$ and the vectors $u_i$ can for instance be obtained by a principal component analysis (PCA).

\subsubsection{Definition and characterization of multiple yield curve models}

We are now in a position to give the definition of a multiple yield curve model.

\begin{definition}\label{def:multcurve}
 Let the number of different tenors be $m:=| \mathcal{D}|$. We call a model consisting of
\begin{itemize}
\item an $\mathbb{R}^{d+n+1}$-valued It\^o-semimartingale $(X,Y,B)$, 
\item vectors  $u_{1}, \ldots, u_{m}$ in $\mathbb{R}^n$, 
\item functions $f_0$, $\eta_0^1, \ldots, \eta_0^m$, 
\item processes $\widetilde{\alpha},\alpha^1, \ldots, \alpha^m$ and $\widetilde{\sigma}, \sigma^1, \ldots, \sigma^m$  
\end{itemize}
an \emph{HJM-type multiple yield curve model} for $\{ (B(t,T))_{t\in[0,T]}\text{ and }(S^{\delta}(t,T))_{t\in[0,T]}, T\geq 0, \delta \in \mathcal{D}\}$ if
\begin{enumerate}
\item $(B, f_0, \widetilde{\alpha}, \widetilde{\sigma},X)$ is a bond price model (as of Definition~\ref{def:bond_price_model});
\item $(u_i^{\top}Y, \eta_0^i,\alpha^i, \sigma^i, X)$ is an HJM-type model (as of Definition~\ref{def:HJMtype}) for $\{ (S^{\delta_i}(t,T))_{t\in[0,T]}, T\geq 0\}$,  for every $i \in \{1, \ldots, m\}$.
\end{enumerate}
\end{definition}

As before, we write $\Sigma^i_t(T)=\int_t^T \sigma_t^i(u)du$ for all $t \leq T$ , $T\geq0$ and $i \in \{1, \ldots, m\}$.
In view of Lemma~\ref{lem:spreadQT}, we define the \emph{risk neutrality} of an HJM-type multiple yield curve model as follows.

\begin{definition}\label{def:riskneutral}
An HJM-type multiple yield curve model is said to be \emph{risk neutral} if
\begin{enumerate}
\item discounted OIS bond prices $\bigl\{ \left(B(t,T)/B_t\right)_{t \in [0,T]}, \, T \geq 0 \bigr\}$ are $\mathbb{Q}$-martingales;
\item for every $T\geq0$, $\{(S^{\delta}(t,T))_{t\in[0,T]}, \delta \in \mathcal{D}\}$ are $\mathbb{Q}^T$-martingales.
\end{enumerate}
\end{definition}
 
The subsequent theorem follows from Theorem~\ref{th:mart} and characterizes condition (ii) of the above definition (recall that condition (i) has been already characterized in Proposition~\ref{prop:bondHJM}).

\begin{theorem}\label{th:multicurve}
For an HJM-type multiple yield curve model satisfying condition (i) of Definition~\ref{def:riskneutral}, the following are equivalent:
\begin{enumerate}
\item condition (ii) of Definition~\ref{def:riskneutral} is satisfied;
\item for every $T \geq 0$ and every $i \in \{1, \ldots, m\}$, the following conditional expectation hypothesis holds:
\begin{align*}
\mathbb{E}^{\mathbb{Q}^{T}}\Bigl[e^{u_i^{\top}Y_T} \,\bigr| \, \mathcal{F}_t\Bigr]
&=e^{u_i^{\top}Y_t+ \int_t^T \eta^i_t(u)du},
\qquad\text{ for all }t\in[0,T];
\end{align*}
\item for every $T\geq0$ and $i \in \{1, \ldots, m\}$, $\bigl(u_i^{\top}, \Sigma^{i\top}- \widetilde{\Sigma}^{\top}\bigr)^{\top}  \in \mathcal{U}^{Y,X}$ and the following conditions are satisfied:
\begin{itemize}
\item the process
\begin{align}\label{eq:martspreadQ}
\Big(\exp \Big(u_i^{\top}Y_t+\int_0^t\bigl(\Sigma^i_s(T)-\widetilde{\Sigma}_s(T)\bigr) dX_s- \int_0^t
\Psi_s^{Y,X}\left( \bigl(u_i^{\top}, \Sigma^{i\top}_s(T)-\widetilde{\Sigma}^{\top}_s(T)\bigr)^{\top}\right)
ds \Big)\Big)_{t\in[0,T]}
\end{align}
is a $\mathbb{Q}$-martingale, for every $T \geq 0$ and $i \in \{1, \ldots, m\}$;
\item the following consistency condition holds for every $i \in \{1, \ldots, m\}$:
\begin{equation}	\label{eq:consistencyspreadQ}
\Psi_t^{Y}(u_i)=\eta^i_{t-}(t),
\qquad\text{ for all }t>0;
\end{equation}
 \item the following HJM drift condition 
\begin{equation}	\label{eq:drift_cond_spread}
\int_t^T \alpha^i_t(u) du=
\Psi_t^{Y}(u_i)-\Psi_t^{Y,X}\left(\bigl(u_i^{\top}, \Sigma^{i\top}_t(T)-\widetilde{\Sigma}^{\top}_t(T)\bigr)^{\top}\right) +\Psi_t^{X}\bigl(-\widetilde{\Sigma}_t(T)\bigr)
\end{equation}
holds for every $t \in [0,T]$, $ T \geq 0$ and $i \in \{1, \ldots, m\}$.
\end{itemize}
\end{enumerate}
\end{theorem}

\begin{proof}
The equivalence between (i) and (ii) can be shown as in the proof of Theorem~\ref{th:mart}. Concerning (iii), note that (i) is -- by Bayes' Theorem -- equivalent to 
\begin{align}\label{eq:spreadQ}
S^{\delta_i}(t,T)\frac{B(t,T)}{B_t }=e^{u_i^{\top}Y_t-\int_0^t r_s ds + \int_t^{T} (\eta^i_t(u)-f_t(u))du}
\end{align}
being a $\mathbb{Q}$-martingale, for every $T\geq0$ and $i=1,\ldots,m$.  Theorem~\ref{th:mart} then yields the following consistency condition
\[
\Psi_t^{u_i^{\top}Y-\int_0^{\cdot} r_s ds}(1)=\eta^i_{t-}(t)-f_{t-}(t),
\qquad\text{ for all }t>0.
\]
Since $\Psi_t^{u_i^{\top}Y-\int_0^{\cdot} r_s ds}(1)=\Psi_t^Y(u_i)-r_{t-}$ and since $r_{t-}=f_{t-}(t)$ by condition (i) of Definition~\ref{def:riskneutral}, condition \eqref{eq:consistencyspreadQ} follows. Concerning the drift condition \eqref{eq:drift_cond_spread}, we have by Theorem~\ref{th:mart} 
\[
\int_t^T \bigl(\alpha^i_t(u)-\widetilde{\alpha}_t(u)\bigr)du=\Psi_t^{u_i^{\top}Y-\int_0^{\cdot} r_s ds}(1)-\Psi_t^{u_i^{\top}Y-\int_0^{\cdot} r_s ds,X}\left(\bigl(1, \Sigma_t^{i\top}(T)-\widetilde{\Sigma}^{\top}_t(T)\bigr)^{\top}\right). 
\]
As the right hand side is equal to
\[
\Psi_t^{Y}(u_i)-\Psi_t^{Y,X}\left(\bigl(u_i^{\top}, \Sigma^{i\top}_t(T)-\widetilde{\Sigma}^{\top}_t(T)\bigr)^{\top}\right)
\]
and as $\int_t^T \widetilde{\alpha}_t(u))du=\Psi_t^X(-\widetilde{\Sigma}_t(T))$ (by Proposition~\ref{prop:bondHJM}), the asserted drift condition follows. By the drift and consistency condition,~\eqref{eq:spreadQ} is actually of the form~\eqref{eq:martspreadQ} (up to the constant term $\exp(\int_0^T (\eta_0^i(u)-f_0(u)) du)$, which finally implies the equivalence between (i) and (iii).
\end{proof}

\begin{remark}
The martingale property of~\eqref{eq:martspreadQ} can be assured similarly as in Remark~\ref{rem:martcond}.
\end{remark}

Additionally, if one is interested in modeling ordered spot spreads $1 \leq S^{\delta_1}(t, t) \leq \cdots \leq S^{\delta_m}(t,t)$, this can be easily obtained by considering a process $Y$ taking values is some cone $C \subset \mathbb{R}^n$ and vectors $u_i \in C^*$ such that $0 < u_1 \prec u_2 \prec \cdots \prec u_m$, where $C^*$ denotes the dual cone of $C$ and $\prec$ the order relation thereon. In that context, we have the following corollary.

\begin{corollary}\label{cor:ordered}
Consider a risk neutral HJM-type multiple yield curve model such that $Y$  takes values in a cone $C \subset \mathbb{R}^n$ and $u_i \in C^*$, for $i=1,\ldots,m$, where $C^*$ denotes the dual cone of $C$. Then Requirement~\ref{req:modelreq} is satisfied. Moreover, if $ u_1 \prec u_2 \prec \cdots \prec u_m$, where $\prec$ denotes the partial order of $C^*$, then we have $S^{\delta_1}(t,T) \leq \cdots \leq S^{\delta_m}(t,T)$ for all $t \leq T$ and $T\geq0$.
\end{corollary}

\begin{proof}
Condition (i) of Requirement~\ref{req:modelreq} is satisfied by definition, while condition (ii) follows from the conditional expectation hypothesis, since $S^{\delta_i}(t,T)=\mathbb{E}^{\mathbb{Q}^{T}}[e^{u_i^{\top}Y_T}|\mathcal{F}_t] \geq 1$ for all $t\leq T$ and $T\geq0$, as $u_i^{\top}Y_T \geq 0$ due to the condition on $Y$ and $u_i$. Similarly, if $u_i \prec u_j$ for $i<j$, 
\[
S^{\delta_i}(t,T)=\mathbb{E}^{\mathbb{Q}^{T}}\Bigl[e^{u_i^{\top}Y_T} \,\bigr| \, \mathcal{F}_t\Bigr]
\leq \mathbb{E}^{\mathbb{Q}^{T}}\Bigl[e^{u_j^{\top}Y_T}\,\bigr| \, \mathcal{F}_t\Bigr]=S^{\delta_j}(t,T).
\]
\end{proof}

\begin{remark}
Note that, in the case when $Y$ is not one-dimensional, the ordering of the spreads can vary over time if the vectors $u_i$, $i \in\{1, \ldots,m\}$ are not  ordered. This means that the model can reproduce market situations where the order relations between spreads associated to different tenors change randomly over time. Nevertheless, $S^{\delta_i}(t,T)\geq 1$ for all $t\leq T$ and $T\geq0$ as long as $u_i \in C^*$.
\end{remark}

\begin{remark}\label{rem:short_spread}
One possibility to specify the process $Y$ is an analogy to the bank account $B=\exp(\int_0^{\cdot}r_sds)$. 
Indeed, let $q$ be an $\RR^n$-valued process and set
\[
Y:=\int_0^{\cdot} q_s ds.
\]
Then,
\[
S^{\delta_i}(t,T)=\mathbb{E}^{\mathbb{Q}^T} \left[e^{\int_0^T u_i^{\top}q_s ds}\bigr|\mathcal{F}_t\right]=e^{\int_0^t u_i^{\top}q_s ds+\int_t^T \eta^{i}_t(u)du},
\]
for all $t\leq T$ and $T\geq0$.
The  consistency condition $\Psi_t^{\int_0^{\cdot} q_s ds}(u_i)=\eta^i_{t-}(t)$
is  then equivalent to $u_i^{\top}q_{t-}=\eta^i_{t-}(t)$ since $\Psi_t^{\int_0^{\cdot} q_s ds}(u_i)=u_i^{\top}q_{t-}$ and the
drift condition becomes
\begin{align*}\int_t^T \alpha^i_t(u) du&=-\Psi_t^{X}\left(\bigl( \Sigma^{i\top}_t(T)-\widetilde{\Sigma}^{\top}_t(T)\bigr)^{\top}\right) +\Psi_t^{X}\bigl(-\widetilde{\Sigma}^{\top}_t(T)\bigr).
\end{align*}
\end{remark}


\section{Construction of risk-neutral HJM-type multiple yield curve models}\label{sec:construction}

Theorem~\ref{th:multicurve} gives necessary and sufficient conditions for an HJM-type multiple yield curve model to be risk neutral. In this section, we provide a general approach to construct risk neutral HJM-type multiple yield curve models starting from a given tuple of basic building blocks, as precisely defined below (see Definition~\ref{def:BB}). 

In particular, we aim at constructing models that satisfy Requirement~\ref{req:modelreq} and can potentially generate spreads which are ordered with respect to the tenor's length. To this effect, if the process $Y$ takes values in some cone $C\subset\RR^n$, part (ii) of Requirement~\ref{req:modelreq} and ordered spreads can then be easily achieved by relying on Corollary~\ref{cor:ordered}. Therefore, the crucial issue is to construct a model that satisfies the three conditions in part (iii) of Theorem~\ref{th:multicurve}, which in particular imply that the model ingredients in Definition~\ref{def:multcurve} cannot be chosen arbitrarily. Our construction starts from the following basic building blocks on a given filtered probability space $(\Omega,\cF,(\cF_t)_{t\geq0},\QQ)$ (compare also with \cite[Definition 4.2]{KK13}). For simplicity of notation, let us denote $u_0:=0\in\RR^n$ and $\Sigma^0_{\cdot}(\cdot):=0\in\RR^n$.

\begin{definition}	\label{def:BB}
A tuple $(X,\Yhat,u_1,\ldots,u_m,f_0,\eta^1_0,\ldots,\eta^m_0,\widetilde{\sigma},\sigma^1,\ldots,\sigma^m)$ is called \emph{building blocks} for an HJM-type multiple yield curve model if
\begin{enumerate}
\item
$(X,\Yhat)$ is an $(\RR^d\times C)$-valued It\^o-semimartingale such that $\Yhat$ is exponentially special and $\Yhat^{\parallel}=\Yhat$, with $\Yhat^{\parallel}$ denoting the dependent part of $\Yhat$ relative to $X$ (see Appendix~\ref{appendix:local_ind});
\item
$u_1,\ldots,u_m$ are vectors in $C^*$, with $C^*$ denoting the dual cone of $C$;
\item
$f_0,\eta^1_0,\ldots,\eta^m_0$ are Borel-measurable functions satisfying condition (ii) of Definition~\ref{def:HJMtype};
\item
$\widetilde{\sigma},\sigma^1,\ldots,\sigma^m$ are $\mathcal{P}\otimes\mathcal{B}(\RR_+)$-measurable processes satisfying condition (iii) of Definition~\ref{def:HJMtype};
\item 
$\bigl(u_i^{\top},\Sigma^{i\top}(T)-\widetilde{\Sigma}^{\top}(T)\bigr)^{\top}\in\mathcal{U}^{\Yhat,X}$, for every $T\geq0$ and $i\in\{0,1,\ldots,m\}$;
\item
the process
\begin{align}\label{eq:martYX}
\left(\exp \left(u_i^{\top}\Yhat_t+\int_0^t\bigl(\Sigma^i_s(T)-\widetilde{\Sigma}_s(T)\bigr) dX_s- \int_0^t
\Psi^{\Yhat,X}_s\Bigl(\bigl(u_i^{\top}, \Sigma^{i\top}_s(T)-\widetilde{\Sigma}^{\top}_s(T)\bigr)^{\top}\Bigr)
ds \right)\right)_{t\in[0,T]}
\end{align}
is a $\QQ$-martingale, for all $T\geq0$ and $i\in\{0,1,\ldots,m\}$.
\end{enumerate}
\end{definition}

Note that, if $(X,\Yhat)$ is chosen to be a L\'evy process (as is the case in Section~\ref{sec:exist_uniq}), condition (vi) of the above definition is automatically satisfied if the volatilities $\widetilde{\sigma},\sigma^1,\ldots,\sigma^m$ are deterministic (more generally, the validity of condition (vi) can be established analogously as in Remark~\ref{rem:martcond}). 

As in the classical HJM framework, the drift processes $(\widetilde{\alpha},\alpha^1,\ldots,\alpha^m)$ will be entirely determined by the building blocks (see part (ii) of Definition~\ref{def:compatible}). Therefore, in view of Definition~\ref{def:BB}, the main model construction problem becomes  finding an $\RR^n$-valued It\^o-semimartingale $Y$ which, together with the given building blocks, generates a risk neutral HJM-type multiple yield curve model, as formalized below.
In particular, note that the process $Y$ needs to satisfy the crucial consistency condition~\eqref{eq:consistencyspreadQ}.

\begin{definition}	\label{def:compatible}
A $C$-valued It\^o-semimartingale $Y$ is said to be \emph{compatible} with the building blocks $(X,\Yhat,u_1,\ldots,u_m,f_0,\eta^1_0,\ldots,\eta^m_0,\widetilde{\sigma},\sigma^1,\ldots,\sigma^m)$ if the following hold:
\begin{enumerate}
\item
$Y^{\parallel}=\Yhat$, with $Y^{\parallel}$ denoting the dependent part of $Y$ relative to $X$;
\item
the tuple $(X,Y,\exp(\int_0^{\cdot}f_s(s)ds),u_1,\ldots,u_m,f_0,\eta^1_0,\ldots,\eta^m_0,\widetilde{\alpha},\alpha^1,\ldots,\alpha^m,\widetilde{\sigma},\sigma^1,\ldots,\sigma^m)$ is a risk-neutral HJM-type multiple yield curve model, in the sense of Definitions~\ref{def:multcurve}-\ref{def:riskneutral}, where
\begin{align}
\widetilde{\alpha}_t(T) &= \partial_T\Psi^{X}_t\bigl(-\widetilde{\Sigma}_t(T)\bigr),
\label{eq:drift_1}\\
\alpha^i_t(T) &= -\partial_T\Psi^{\Yhat,X}_t\Bigl(\bigl(u_i^{\top},  \Sigma_t^{i\top}(T)-\widetilde{\Sigma}_t^{\top}(T)\bigr)^{\top}\Bigr)+\partial_T\Psi^{X}_t\bigl(-\widetilde{\Sigma}_t(T)\bigr),
\label{eq:drift_2}
\end{align}
for all $t\leq T$, $T\geq0$ and $i\in\{1,\ldots,m\}$.
\end{enumerate}
\end{definition}

In other words, starting from given building blocks $(X,\Yhat,u_1,\ldots,u_m,f_0,\eta^1_0,\ldots,\eta^m_0,\widetilde{\sigma},\sigma^1,\ldots,\sigma^m)$ and then searching for a compatible It\^o-semimartingale $Y$ amounts to a model construction strategy which proceeds along three subsequent steps:
\begin{itemize}
\item[(a)]
define the drift processes $(\widetilde{\alpha},\alpha^1,\ldots,\alpha^m)$ via the right-hand sides of \eqref{eq:drift_1}-\eqref{eq:drift_2};
\item[(b)]
prove the existence and the uniqueness of the generalized forward rate processes $(f,\eta^1,\ldots,\eta^m)$, given as the solutions to \eqref{eq:eta} with initial values $(f_0,\eta^1_0,\ldots,\eta^m_0)$, drift processes $(\widetilde{\alpha},\alpha^1,\ldots,\alpha^m)$ and volatility processes $(\widetilde{\sigma},\sigma^1,\ldots,\sigma^m)$;
\item[(c)]
construct a $C$-valued It\^o-semimartingale $Y$ satisfying the three following requirements:
\begin{itemize}
\item[(i)] $Y^{\parallel}=\Yhat$;
\item[(ii)] $\Psi_t^{Y}(u_i)=\eta_{t-}^i(t)$, for all $t>0$ and $i\in\{1,\ldots,m\}$ (\emph{consistency condition});
\item[(iii)] the process given in \eqref{eq:martspreadQ} is an $\QQ$-martingale, for every $T\geq0$ and $i=1,\ldots,m$.
\end{itemize}
\end{itemize}
If steps (b)-(c) can be successfully solved, then a risk-neutral HJM-type multiple yield curve model is given by the tuple $(X,Y,\exp(\int_0^{\cdot}f_s(s)ds),u_1,\ldots,u_m,f_0,\eta^1_0,\ldots,\eta^m_0,\widetilde{\alpha},\alpha^1,\ldots,\alpha^m,\widetilde{\sigma},\sigma^1,\ldots,\sigma^m)$. Indeed, in view of Theorem~\ref{th:multicurve}, the HJM drift condition \eqref{eq:drift_cond_spread} follows from step (a), noting that
\[
-\partial_T\Psi^{\Yhat,X}_t\Bigl(\bigl(u_i^{\top},  \Sigma_t^{i\top}(T)-\widetilde{\Sigma}_t^{\top}(T)\bigr)^{\top}\Bigr)
= -\partial_T\Psi^{Y,X}_t\Bigl(\bigl(u_i^{\top},  \Sigma_t^{i\top}(T)-\widetilde{\Sigma}_t^{\top}(T)\bigr)^{\top}\Bigr),
\]
for all $t\leq T$, $T\geq0$ and $i\in\{1,\ldots,m\}$, since $Y^{\parallel}=\Yhat$  together with Definition~\ref{def:local_independence} and Lemma~\ref{lem:local_independence} implying local independence of $Y^{\perp}:=Y-Y^{\parallel}$ and $(Y^{\parallel},X)$. The consistency condition and the martingale property of the process in \eqref{eq:martspreadQ} follow from step (c). Finally, part (ii) of Requirement~\ref{req:modelreq} and ordered spreads can be achieved by taking, for instance, $C=\RR_+$, as considered in Section~\ref{sec:constr_Yperp}.

From now on, we fix a given tuple of building blocks $(X,\Yhat,u_1,\ldots,u_m,f_0,\eta^1_0,\ldots,\eta^m_0,\widetilde{\sigma},\sigma^1,\ldots,\sigma^m)$. In Section~\ref{sec:exist_uniq}, we prove the existence and the uniqueness of the forward curves $(f,\eta^1,\ldots,\eta^m)$, thus solving step (b) above, while, in Section~\ref{sec:constr_Yperp}, we present a general procedure to construct a compatible It\^o-semimartingale $Y$, thus solving step (c) above.

\subsection{Existence and uniqueness of the forward spread curves}	\label{sec:exist_uniq}

In order to address the issue of existence and uniqueness of $\eta^i$, $i \in \{1,\ldots,m\}$, and also of the OIS forward curve $f$, we shall rely on the results of~\cite{FTT:10}, adapted to the present multiple curve setting. For notational convenience, we denote the OIS forward curve by $\eta^0:=f$ and its volatility and  drift by  $\sigma^0_t(T):=\widetilde{\sigma}_t(T)$ and $\alpha^0_t(T):=\widetilde{\alpha}_t(T)$.

We are interested in volatility structures which depend on the forward (spread) curves $\eta_t(\cdot):=\bigl(\eta_t^0(\cdot),\eta_t^1(\cdot),\ldots,\eta_t^m(\cdot)\bigr)^{\top}$ in the following way
\begin{align}\label{eq:sigmamus}
\sigma^i_t(T)=
\left\{ \begin{array}{ll}
\zeta^i(\theta_t)(T-t), & t \leq T,\\
0, & t > T.
\end{array} \right., \quad i \in \{0,1, \ldots, m\},
\end{align}
where $\theta_t(s):=\eta_t(t+s)$ corresponds to the Musiela parametrization and $\zeta^i$, for $i\in\{0,1,\ldots,m\}$, is a function from some Hilbert space $H^{\lambda}_{m+1}$ of forward (spread) curves $h: \mathbb{R}_+ \to \mathbb{R}^{m+1}$ specified below.
Note that, for all $i\in\{0,1,\ldots,m\}$, the volatility $\sigma^i_t(T)$ of each individual forward curve $\eta^i$ is allowed to depend through the function $\zeta^i$ on the whole family of forward curves $\eta=(\eta^0,\eta^1,\ldots,\eta^m)^{\top}$. In view of the empirical facts reported in Section \ref{subsec:market}, this is a relevant feature of the model.
We therefore switch to the Musiela parametrization and view $(\theta_t)_{t\geq 0}$ as a single $H^{\lambda}_{m+1}$-valued stochastic process. 

Until the end of Section~\ref{sec:exist_uniq}, in order to apply the results of~\cite{FTT:10}, we assume that $(X,\Yhat)$ is a L\'evy martingale taking values in $\RR^{d+n}$. Without loss of generality, the driving semimartingale $(X_t)_{t\geq0}$ is then of the form
\[
X_t=\beta_t+ \int_0^t \int_{\mathbb{R}^d} \xi \bigl(\mu(dt,d\xi)-F(d\xi)dt\bigr),
\]
where $(\beta_t)_{t \geq 0}$ is an $\mathbb{R}^d$-valued standard Brownian motion 
and $\mu$ a homogeneous Poisson random measure on $\mathbb{R}_+ \times \mathbb{R}^d$ with compensator $F(d\xi)dt$. 
The SDE for $\eta^i$, for $i \in \{0,1, \ldots, m\}$, thus becomes
\begin{align}
\eta^i_t(T)&=\eta^i_0(T)+\int_0^t \alpha^i_s(T)ds+\int_0^t \sigma^i_s(T) dX_s \notag\\
&=\eta^i_0(T)+\int_0^t \alpha^i_s(T)ds+\int_0^t \sigma^i_s(T) d\beta_s+\int_0^t \int_{\mathbb{R}^d}  \bigl(\sigma^i_s(T)\bigr)^{\top}\xi \bigl(\mu(ds,d\xi)-F(d\xi)ds\bigr).\label{eq:SDEeta}
\end{align}
Note that the 
processes $\sigma(t,T)$ and $\gamma(t,\xi,T)$ in~\cite{FTT:10} correspond to $\sigma^i_t(T)$ and $\bigl(\sigma^{i}_t(T)\bigr)^{\top}\xi$, respectively, for $i\in \{0,1, \ldots, m\}$.
Assuming continuity of $T \mapsto \eta_t(T)$, we can transform~\eqref{eq:SDEeta} into the following integral equation for $\theta^i$:
\begin{equation}
\begin{split}\label{eq:SPDE}
\theta^i_t(x)&=S_t\eta^i_0(x)+\int_0^t S_{t-s}\alpha^i_s(s+x)ds+\int S_{t-s}\sigma^i_s(s+x) d\beta_s\\
&\quad +\int_0^t \int_{\mathbb{R}^d} S_{t-s}\bigl(\sigma^i_s(s+x)\bigr)^{\top}\xi \bigl(\mu(ds,d\xi)-F(d\xi)ds\bigr), \quad i \in \{0,1, \ldots, m\},
\end{split}
\end{equation}
where $(S_t)_{t \geq 0}$ denotes the shift semigroup, that is $S_th=h(t+\cdot)$.
In order to establish existence of solutions to such equations, let us introduce the following spaces of forward curves, in line with~\cite{FTT:10} (but generalized to the multivariate case). Fix an arbitrary constant $\lambda >0$ and let $H_k^{\lambda}$ be the space of all absolutely continuous functions $h: \mathbb{R}_+ \to \mathbb{R}^k$ such that
\[
\| h\|_{\lambda,k}:=\left(\| h(0)\|_{k}^2 + \int_{\mathbb{R}_+} \|\partial_s h(s)\|_{k}^2e^{\lambda s} ds\right)^{\frac{1}{2}},
\]
where $\|\cdot\|_k$ denotes the norm in $\mathbb{R}^k$, with $k\in\{1,d,m+1\}$.

As stated above, we shall consider drift and volatility structures which are functions of the prevailing forward (spread) curves, i.e., 
\begin{align*}
\alpha^i_t(T)=
\left\{ \begin{array}{ll}
\kappa^i(\theta_t)(T-t), & t \leq T,\\
0, & t > T,
\end{array} \right.,\quad
\sigma^i_t(T)=
\left\{ \begin{array}{ll}
\zeta^i(\theta_t)(T-t), & t \leq T,\\
0, & t > T,
\end{array} \right.,
\end{align*}
for all $i \in \{0,1, \ldots,m\}$.
In particular, we require $\kappa^i: H^{\lambda}_{m+1} \to H^{\lambda}_1$ and $\zeta^i:  H^{\lambda}_{m+1} \to H^{\lambda}_d$. 

Let us denote by $c^{\Yhat,X}$ and $K^{\Yhat,X}$ the second and third terms of the L\'evy triplet of $(\Yhat,X)$, so that $c^{\Yhat,X}\in\RR^{n\times d}$ and $K^{\Yhat,X}$ is a L\'evy measure on $\mathbb{R}^{n\times d}$, with $X$-marginal denoted by $F(d\xi)$.
The drift conditions~\eqref{eq:drift_1}-\eqref{eq:drift_2} then read as 
\begin{align*}
\alpha^i_t(T)&=-u_i^{\top}c^{\Yhat,X}\bigl(\sigma_t^i(T)-\sigma_t^0(T)\bigr)- \bigl(\sigma_t^i(T)-\sigma_t^0(T)\bigr)^{\top}\bigl(\Sigma_t^i(T)-\Sigma_t^0(T)\bigr)\\
&\quad - \int \bigl(\sigma_t^i(T)-\sigma_t^0(T)\bigr)^{\top}\xi\bigl(e^{u_i^{\top} \xihat+ (\Sigma_t^i(T)-\Sigma_t^0(T))^{\top}\xi}-1\bigr)K^{\Yhat,X}(d\xihat, d\xi)\\
&\quad + \bigl(\sigma_t^0(T)\bigr)^{\top}\Sigma_t^0(T)-\int \sigma_t^0(T)^{\top}\xi \bigl(e^{-(\Sigma_t^0(T))^{\top}\xi}-1\bigr)F(d\xi), \qquad  i \in \{0,1, \ldots,m\},
\end{align*}
as long as  
\begin{align*}
& \int \sup_{T \geq t}\left(\bigl(\sigma_t^i(T)-\sigma_t^0(T)\bigr)^{\top}\xi\bigl(e^{u_i^{\top} \xihat+ (\Sigma_t^i(T)-\Sigma_t^0(T))^{\top}\xi}-1\bigr)\right)K^{\Yhat,X}(d\xihat, d\xi) < \infty,\\
&\int \sup_{T \geq t}\left(\bigl(\sigma_t^0(T)\bigr)^{\top}\xi \bigl(e^{-(\Sigma_t^0(T))^{\top}\xi}-1\bigr)\right)F(d\xi)< \infty,
\end{align*} 
so that we are allowed to differentiate under the integral sign. 
This translates to $\kappa^i$ as follows, for all $h\in H^{\lambda}_{m+1}$, 
\begin{equation}\label{eq:kappa}
\begin{split}
\kappa^i(h)(s)&=-u_i^{\top}c^{\Yhat,X}\bigl(\zeta^i(h)(s)-\zeta^0(h)(s)\bigr)- \bigl(\zeta^i(h)(s)-\zeta^0(h)(s)\bigr)^{\top}\bigl(Z^i  (h)(s)-Z^0 (h)(s)\bigr)\\
&\quad - \int \bigl(\zeta^i(h)(s)-\zeta^0(h)(s)\bigr)^{\top}\xi(e^{u_i^{\top} \xihat+ (Z^i  (h)(s)-Z^0 (h)(s))^{\top}\xi}-1)K^{\Yhat,X}(d\xihat, d\xi)\\
&\quad + \bigl(\zeta^0(h)(s)\bigr)^{\top}Z^0 (h)(s)-\int \bigl(\zeta^0(h)(s)\bigr)^{\top}\xi (e^{-(Z^0 (h)(s))^{\top}\xi}-1)F(d\xi),
\quad i\in\{0,1,\ldots,m\},
\end{split}
\end{equation}
where $Z^i(h)(s):=\int_0^s \zeta^i(h)(u)du$. In the sequel, for a function $g:H^{\lambda}_{m+1} \to H^{\lambda}_{d}$ and a vector $z\in\mathbb{R}^d$, we shall write $g(h)^{\top}z$ for $\sum_{j=1}^d z_j g_j(h)$.
The above specification leads to forward (spread) rates in ~\eqref{eq:SPDE} being a solution of 
\begin{equation}	\label{eq:SPDE1}
\theta^i_t=S_t\eta^i_0+\int_0^t S_{t-s}\kappa^i(\theta_s)ds+\int S_{t-s}\zeta^i(\theta_s) d\beta_s
+\int_0^t \int_{\mathbb{R}^d} S_{t-s}\bigl(\zeta^i(\theta_{s-})\bigr)^{\top}\xi \bigl(\mu(ds,d\xi)-F(d\xi)ds\bigr), 
\end{equation}
for $i\in\{0,1,\ldots,m\}$ and  where $\kappa$ is specified in~\eqref{eq:kappa}. 
An $H^{\lambda}_{m+1}$-valued process $\theta$ satisfying~\eqref{eq:SPDE1} is said to be a \emph{mild solution} to the stochastic partial differential equation (for $i\in\{0,1,\ldots,m\}$)
\begin{align}\label{eq:mildSPDE}
d\theta^i_t=\left(\frac{d}{ds}\theta^i_t+\kappa^i(\theta_t)\right)dt+\zeta^i(\theta_t) d\beta_t+\int_{\mathbb{R}^d}\bigl(\zeta^i(\theta_{t-})\bigr)^{\top}\xi \bigl(\mu(dt,d\xi)-F(d\xi)dt\bigr), \quad \theta^i_0=\eta^i_0.
\end{align}
We are thus concerned with the question of existence of mild solutions to~\eqref{eq:mildSPDE}. Following~\cite{FTT:10}, such SPDEs can be understood as time-dependent transformations of time-dependent SDEs with infinite dimensional state space (for more details, see~\cite[Equation 1.11]{FTT:10}). 

For convenience of notation, we denote $\zeta^{i0}(h):=\zeta^i(h)-\zeta^0(h)$ and  $Z^{i0}(h):=Z^i(h)-Z^0(h)$, for all $i\in\{1,\ldots,m\}$ and $h\in H^{\lambda}_{m+1}$, and decompose $\kappa^i(h)=\kappa_1^i(h)+\kappa_2^i(h)+\kappa_3^i(h)+\kappa_4(h)+\kappa_5(h)$, 
where
\begin{align*}
\kappa_1^i(h)&=-u_i^{\top}c^{\Yhat,X}\bigl(\zeta^i(h)-\zeta^0(h)\bigr)
=-u_i^{\top}c^{\Yhat,X}\bigl(\zeta^{i0}(h)\bigr),\\
\kappa_2^i(h)&=-\bigl(\zeta^i(h)-\zeta^0(h)\bigr)^{\top}(Z^i  (h)-Z^0 (h))
= -\bigl(\zeta^{i0}(h)\bigr)^{\top}Z^{i0}(h), \\
\kappa_3^i(h)&=- \int \bigl(\zeta^i(h)-\zeta^0(h)\bigr)^{\top}\xi\bigl(e^{u_i^{\top} \xihat+ (Z^i  (h)-Z^0 (h))^{\top}\xi}-1\bigr)K^{\Yhat,X}(d\xihat,d\xi)\\
&=- \int \bigl(\zeta^{i0}(h)\bigr)^{\top}\xi\bigl(e^{u_i^{\top} \xihat+ (Z^{i0}(h))^{\top}\xi}-1\bigr)K^{\Yhat,X}(d\xihat,d\xi), \\
\kappa_4(h)&=\bigl(\zeta^0(h)\bigr)^{\top}Z^0 (h), \\
\kappa_5(h)&=-\int \bigl(\zeta^0(h)\bigr)^{\top}\xi \bigl(e^{-(Z^0 (h))^{\top}\xi}-1\bigr)F(d\xi).
\end{align*}

Aiming at establishing existence and uniqueness of a mild solution to~\eqref{eq:mildSPDE}, let us introduce suitable growth and Lipschitz continuity conditions on the volatility functions $\zeta^i$, for all $i=0,1,\ldots,m$, as formulated in the following assumption (compare also with~\cite[Assumption 3.1]{FTT:10}).

\begin{assumption}\label{ass:HJMMcond}
$\zeta^i: H^{\lambda}_{m+1} \to H^{\lambda,0}_d$, for all $i=0,1,\ldots,m$, where $H^{\lambda,0}_k:=\{ h \in H^{\lambda}_k\, | \, \|h(\infty)\|_{k}=0\}$, for $k\in\{1,d\}$.
Moreover, for all $i \in \{0,1, \ldots,m\}$, there exist positive constants $C_i$, $L_i$, $M_i$ such that
\begin{align*}
&\|Z^{i}(h)(s)\|_d \leq C_i,
\qquad\qquad\qquad\qquad\qquad\qquad 
\text{for all }h \in H^{\lambda}_{m+1}, s\in \mathbb{R}_+,\\
&\|\zeta^i(h_1)-\zeta^{i}(h_2)\|_{\lambda, d} \leq L_i\| h_1 -h_2 \|_{\lambda, m+1}, 
\qquad 
\text{for all }h_1,h_2 \in H^{\lambda}_{m+1},\\
&\|\zeta^i(h)\|_{\lambda, d}\leq M_i,
\qquad\qquad\qquad\qquad\qquad\qquad\;\; 
\text{for all }h \in H^{\lambda}_{m+1},\\
\end{align*}
and constants $K_0>0$ and $K_i>0$ such that
\begin{align}
&\int e^{C_0\|\xi\|_d} \bigl(\|\xi\|_d^2 \vee \|\xi\|_d^4\bigr) F(d\xi) \leq K_0, \label{eq:ass4}\\
&\int e^{\|u_i\|_n\|\xihat\|_n +(C_0+C_i)\|\xi\|_d} \bigl(\|\xihat\|_n^2+(\|\xi\|_d^2 \vee \|\xi\|_d^4) \bigr) K^{\Yhat,X}(d\xihat, d\xi) \leq K_i, \quad i\in \{1, \ldots,m\}.\label{eq:ass5}
\end{align}
Furthermore, we suppose that, for each $h \in H^{\lambda}_{m+1}$, the maps $\kappa_3^i(h)$ and $\kappa_5(h)$
are absolutely continuous with weak derivatives
\begin{align}
\frac{d}{ds}\kappa^i_3(h)&=-\int \bigl((\zeta^{i0}(h))^{\top}\xi\bigr)^2\bigl(e^{u_i^{\top} \xihat+ (Z^{i0} (h))^{\top}\xi}\bigr)K^{\Yhat,X}(d\xihat, d\xi)\label{eq:derivativekappa3}\\
&\quad -\int \left(\frac{d}{ds}\bigl(\zeta^{i0}(h)\bigr)^{\top}\xi\right)\bigl(e^{u_i^{\top} \xihat+ (Z^{i0}
(h))^{\top}\xi}-1\bigr)K^{\Yhat,X}(d\xihat, d\xi),\notag\\
\frac{d}{ds}\kappa_5(h)&=\int \bigl((\zeta^0(h))^{\top}\xi\bigr)^2 e^{-(Z^0 (h))^{\top}\xi}F(d\xi)-\int \frac{d}{ds}\bigl(\zeta^0(h)\bigr)^{\top}\xi\bigl(e^{-(Z^0 (h))^{\top}\xi}-1\bigr)F(d\xi).	\label{eq:derivativekappa5}
\end{align}
\end{assumption}

As shown in the next proposition (the rather technical proof of which is postponed to Appendix~\ref{app:proof}), Assumption~\ref{ass:HJMMcond} implies that the drift functions $\kappa^i$, for all $i=0,1,\ldots,m$, are also Lipschitz continuous. This property will be crucial in order to establish existence and uniqueness of a mild solution to \eqref{eq:mildSPDE}.

\begin{proposition}\label{prop:Lip}
Suppose that Assumption~\ref{ass:HJMMcond} is satisfied. Then, for all $i \in \{0,1, \ldots, m\}$, it holds that  $\kappa^{i}(H^{\lambda}_{m+1})\subseteq H^{\lambda,0}_1$ and there exist constants $Q_i>0$ such that 
\begin{align}\label{eq:kappaLip} 
\| \kappa^i(h_1)-\kappa^i(h_2)\|_{\lambda,1} \leq Q_i \| h_1 -h_2 \|_{\lambda, m+1}
\end{align}
for all $h_1, h_2 \in H^{\lambda}_{m+1}$.
\end{proposition}

We are now in a position to state the following theorem, which asserts existence and uniqueness of a mild solution to~\eqref{eq:mildSPDE} and extends~\cite[Theorem 3.2]{FTT:10} to the present multiple curve setting.

\begin{theorem}	\label{thm:existence_uniqueness}
Suppose that Assumption~\ref{ass:HJMMcond} is satisfied. Then, for each initial curve $\theta_0 \in  H_{m+1}^{\lambda}$, 
there exists a unique adapted c\`adl\`ag, mean-square continuous mild  $H^{\lambda}_{m+1}$-valued solution $(\theta_t)_{t\geq 0}$  satisfying 
\[
\mathbb{E}\biggl[\sup_{t\in[0,T]}\|\theta_t\|^2_{\lambda,m+1}\biggr] < \infty, 
\quad \text{for all }T \in \mathbb{R}_+.
\]
\end{theorem}
\begin{proof}
By virtue of~\cite[Theorem 2.1]{FTT:10}, Assumption~\ref{ass:HJMMcond} and Proposition~\ref{prop:Lip}, ~\cite[Corollary 10.9]{FTT:10b} applies and yields the claimed existence and uniqueness result.
\end{proof}

\begin{remark}
In view of applications, one is often interested in constructing multiple yield curve models producing positive OIS forward rates $f$ as well as positive forward spread rates $\eta^i$, for $i\in\{1,\ldots,m\}$. Similarly as in the case of Theorem \ref{thm:existence_uniqueness}, necessary and sufficient conditions for the positiveness of $f$ and $\eta^i$, for $i\in\{1,\ldots,m\}$, can be obtained by adapting to the present context the results of \cite[Section 4]{FTT:10}.
\end{remark}

\subsection{Construction of $Y^{\perp}$}	
\label{sec:constr_Yperp}

Until the end of the present section, we shall suppose existence and uniqueness of the forward curves $(f,\eta^1,\ldots,\eta^m)$, but we do not necessarily assume that $(X,\Yhat)$ is a L\'evy martingale. We now present a general procedure to construct an It\^o-semimartingale $Y$ compatible with the building blocks $(X,\Yhat,u_1,\ldots,u_m,f_0,\eta^1_0,\ldots,\eta^m_0,\widetilde{\sigma},\sigma^1,\ldots,\sigma^m)$, in the sense of Definition~\ref{def:compatible}, or, equivalently, satisfying the three requirements in step (c) of the model construction procedure described at the beginning of Section~\ref{sec:construction}.

As a preliminary observation, note that, by the definition of local independence (see Definition~\ref{def:local_independence} and Lemma~\ref{lem:local_independence}), constructing a $C$-valued process $Y$ such that $Y^{\parallel}=\Yhat$ (requirement (i) of step (c)) can be achieved by constructing a $C$-valued process $Y^{\perp}$ which is locally independent of $(X,\Yhat)$ and then letting $Y:=\Yhat+Y^{\perp}$. By local independence, the local exponent $\Psi^{Y^{\perp}}$ must then satisfy the following condition, which amounts to the consistency condition (requirement (ii) of step (c)):
\begin{equation}	\label{eq:cons_cond_Yperp}
\Psi^{Y^{\perp}}_t(u_i) = \eta^i_{t-}(t)-\Psi^{\Yhat}_t(u_i),
\qquad\text{ for all }t>0\text{ and }i\in\{1,\ldots,m\}.
\end{equation}

In the special case $n=m$, one can arbitrarily choose the characteristics $c^{Y^{\perp}}$ and $K^{Y^{\perp}}$ and then, for all $i=1,\ldots,n$, specify the drift characteristic $b^{Y^{\perp,i}}$ as the predictable process $b^{Y^{\perp,i}}_t=\eta^i_{t-}(t)-\Psi^{\Yhat}_t(u_i)-\frac{1}{2}c^{Y^{\perp},ii}_t-\int\bigl(e^{\xi^i}-1-\chi(\xi)^i\bigr)K^{Y^{\perp}}_t(d\xi)$, where the vectors $\{u_1,\ldots,u_n\}$ are basis vectors in $\RR^N$, so that \eqref{eq:cons_cond_Yperp} holds by construction.
However, the case $m=n$ is rather unrealistic, since we aim to model the log-spot spreads for different tenors by means of a lower dimensional process $Y$ in order to capture their interdependencies. In the latter case (i.e., when $m > n$), working on the drift characteristic does not suffice any more. Note also that, even in the case $n=m$, one has to impose further conditions in order to ensure that $Y^{\perp}$ lies in $C$. 

For simplicity of presentation, we consider the case where $Y$ is a one-dimensional process taking values in the cone $C=\RR_+$ and $0< u_1 <u_2 < \ldots <u_m$. We aim at constructing a process $Y^{\perp}$, locally independent of $(X,\Yhat)$, such that the consistency condition \eqref{eq:cons_cond_Yperp} is satisfied and the process given in equation \eqref{eq:martspreadQ} is a martingale (requirement (iii) of step (c)).
We shall construct the process $Y^{\perp}$ as a finite activity pure jump process (see however Remark \ref{rem:driftdiff}) on a suitably extended probability space. Note that, since we want $Y^{\perp}$ to take values in $\RR_+$, we need to restrict its jump sizes so that $\Delta Y^{\perp}\geq -Y^{\perp}_-$ a.s. Hence, the construction problem amounts to determine the compensating jump measure of $Y^{\perp}$, which we denote as $K_t\bigl(\omega,Y_{t-}^{\perp}(\omega),d\xi\bigr)dt$ in order to make explicit the dependence of the jump size on $Y^{\perp}_-$.

The crucial consistency condition \eqref{eq:cons_cond_Yperp} will be satisfied if the kernel $K_t(\omega,y,d\xi)$ satisfies, for all $\omega\in\Omega$, $y\in\RR_+$ and $t>0$,
\begin{equation}\label{eq:momentproblem}
\int (e^{u_i\xi}-1)K_t(\omega,y,d\xi)=\eta^i_{t-}(t)(\omega)-\Psi_t^{\Yhat}(u_i)(\omega)
 =: p^i_t(\omega),
 \qquad \text{ for }i=1,\ldots,m.
\end{equation}
Note that the right-hand side is fully determined from the previous steps of the model's construction.
In particular, \eqref{eq:momentproblem} means that, for every $\omega \in \Omega$, $y \in \RR_+$ and $t>0$, $K_t(\omega, y, d\xi)$ needs to be a measure on $\bigl(\RR,\mathcal{B}(\RR)\bigr)$ supported on $[-y, \infty)$. 
Moreover, in order to ensure the martingale property of \eqref{eq:martspreadQ}, we also require $K_t(\omega,y,d\xi)$ to satisfy the following integrability condition, for all $\omega\in\Omega$, $y\in\RR_+$ and $t\geq0$:
\begin{equation}	\label{eq:mart_cond}
\int\bigl(|1\vee\xi|\bigr)e^{u_m\xi^+}K_t(\omega,y,d\xi) = p^{m+1}_t(\omega,y),
\end{equation}
with respect to some family $\bigl\{\bigl(p^{m+1}_t(\cdot,y)\bigr)_{t\geq0}\, |\, y\in\RR_+\bigr\}$ of predictable processes, measurable with respect to $y$ and satisfying $p^{m+1}_t(\omega,y)\leq \Hline$ $\QQ$-a.s. for all $y\in\RR_+$ and $t\geq0$, for some constant $\Hline>0$.

For fixed $\omega\in\Omega$, $y\in\RR_+$ and $t\geq0$, the question of whether such a measure  $K_t(\omega,y,\cdot)$ exists corresponds to the generalized moment problem considered by Krein and Nudelman~\cite{KN:77} and puts some restrictions on the possible values of $\bigl(p_t^1(\omega),\ldots,p_t^m(\omega),p_t^{m+1}(\omega,y)\bigr)$.

Let us briefly recall the formulation of the generalized moment problem. Let $[a,b]  \subset \mathbb{R}$ (with $b$ possibly $\infty$) be some interval and consider a family of linearly independent continuous functions $f_i: [a,b] \to \mathbb{R}$, $i=1,\ldots, m+1$. Let $c \in\RR^{m+1}$. Then the generalized moment problem consists in finding a positive measure $\mu$ on $\bigl([a,b], \mathcal{B}([a,b])\bigr)$ such that
\[
\int_a^b f_i(\xi)\mu(d\xi)=c_i,
\qquad\text{for all }i=1,\ldots,m+1.
\]
Under the condition that there exists some function $h$ being a linear combination of $f_i$, $i=1,\ldots, m+1$, which is strictly positive on $[a,b]$, the result of Krein and Nudelman~\cite[Theorem I.3.4, Theorem III 1.1 and p. 175]{KN:77} states that the generalized moment problem admits a solution if and only if $c$ is an element of the closed conic hull $K(U)$ of 
\[
U=\bigl\{\bigl(f_1(\xi),\ldots, f_{m+1}(\xi)\bigr)\,|\, \xi \in [a,b]\bigr\}.
\]
In our context, this directly implies the following lemma. As a preliminary, let us define the family of functions $g_i(\xi):=e^{u_i\xi}-1$, for $i=1,\ldots,m$, and $g_{m+1}(\xi):=\bigl(|\xi|\vee1\bigr)e^{(u_m\vee 1)|\xi|}$.

\begin{lemma}	\label{lem:existencemeasure}
Let $0<u_1 < \ldots < u_m$.
Then, for every $\omega\in\Omega$, $y\in\RR_+$ and $t\geq0$, there exists a non-negative measure $K_t(\omega,y,d\xi)$ on $\bigl([-y,\infty),\mathcal{B}([-y,\infty))\bigr)$ satisfying \eqref{eq:momentproblem}-\eqref{eq:mart_cond} if and only if
\begin{equation}	\label{eq:exsitencemeasure}
p_t(\omega,y) := 
\bigl(p^1_t(\omega),\ldots,p_t^m(\omega),p^{m+1}_t(\omega,y)\bigr) \in 
K\Bigl(\bigl\{\bigl(g_1(\xi),\ldots,g_m(\xi),g_{m+1}(\xi)\bigr)\,|\,\xi\in[-y,\infty)\bigr\}\Bigr).
\end{equation}
\end{lemma}
\begin{proof}
For every fixed  $\omega\in\Omega$, $y\in\RR_+$ and $t\geq0$, the claim follows directly from \cite[Theorem I.3.4, Theorem III 1.1 and p. 175]{KN:77}, noting that the functions $g_i$, $i=1,\ldots,m+1$, are continuous and linearly independent and that the function $g_{m+1}$ is strictly positive. 
\end{proof}

As we are going to show in the remaining part of this section, the construction of a process $Y^{\perp}$ satisfying all the desired properties will be possible as long as there exists a solution to the generalized moment problem. More precisely, in view of Lemma \ref{lem:existencemeasure}, let us formulate the following assumption.

\begin{assumption}	\label{ass:ex_sol_moment}
There exists a family $\bigl\{\bigl(p^{m+1}_t(\cdot,y)\bigr)_{t\geq0}\,|\, y\in\RR_+\bigr\}$ of predictable processes, measurable with respect to $y$, satisfying $p^{m+1}_t(\omega,y)\leq \Hline$ $\QQ$-a.s. for all $y\in\RR_+$ and $t\geq0$, for some constant $\Hline>0$, such that condition \eqref{eq:exsitencemeasure} is satisfied, for all $\omega\in\Omega$, $y\in\RR_+$ and $t\geq0$. 
\end{assumption}

\begin{remark}
If $m=1$, then for any given $p_t^1(\omega)$ and $y>0$, we can always find some $p_t^2(\omega,y)$ such that 
$(p_t^1(\omega), p_t^2(\omega,y))\in K(\{(g_1(\xi), g_2(\xi))\,|\, \xi \in [-y, \infty)\})$. If $y=0$, then $p_t^1(\omega)$ has to be nonnegative. In particular, if $\omega \mapsto p_t^1(\omega)$ is bounded and nonnegative, it is easy to see that Assumption~\ref{ass:ex_sol_moment} is always satisfied.
Similarly, for $m=2$ the conditions $p_t^1(\omega) \geq 0$ and $p_t^2(\omega) \geq \frac{u_2}{u_1}p_t^1(\omega)$ and boundedness (in $\omega$) are sufficient for the validity of Assumption~\ref{ass:ex_sol_moment}.
\end{remark}

The next proposition establishes the existence of a process $Y^{\perp}$ with jump measure $K_t\bigl(\omega,Y^{\perp}_{t-}(\omega),d\xi\bigr)dt$ on an extension of the original probability space $(\Omega, \mathcal{F}, (\mathcal{F}_t)_{t\geq0},\mathbb{Q})$.
We rely on a constructive proof on a specific stochastic basis which is defined as follows (compare also with~\cite[Appendix A]{cfmt11}):
\begin{itemize}
\item $(\widetilde{\Omega}, \mathcal{G}, (\mathcal{G}_t)_{t\geq 0})$ is a filtered space, with $\widetilde{\Omega}:=\Omega \times \Omega'$, $\mathcal{G}_t:=\bigcap_{s>t}\mathcal{F}_s\otimes \mathcal{H}_s$ and $\mathcal{G}=\mathcal{F}\otimes \mathcal{H}$. Here, $(\Omega, \mathcal{F}, (\mathcal{F}_t)_{t\geq0},\mathbb{Q})$ is the probability space on which we worked so far and $(\Omega',\mathcal{H}, \mathcal{H}_t)$ is precisely defined below. Note that we do not assume to have a measure on $(\widetilde{\Omega}, \mathcal{G})$ for the moment. The generic sample element is denoted by $\widetilde{\omega}:=(\omega,\omega')\in \widetilde{\Omega}$.
\item $(\Omega', \mathcal{H})$ is the canonical space of real-valued marked point processes (see e.g.\ \cite{jacod7475}), meaning that $\Omega'$ consists of all c\`adl\`ag piecewise constant functions $\omega': \bigl[0,T_{\infty}(\omega')\bigr) \rightarrow \mathbb{R}$ with $\omega'(0)=0$ and $T_{\infty}(\omega')=\lim_{n\rightarrow\infty} T_n(\omega')\leq \infty$, where $T_n(\omega')$, defined by $T_0=0$ and
\[
  T_n(\omega'):=\inf\bigl\{t > T_{n-1}(\omega')\,|\,
\omega'(t)\neq \omega'(t-)\bigr\}\wedge \infty ,\quad \text{ for }n\ge 1,
\]
are the successive jump times of $\omega'$. We denote by
\[
J_t(\widetilde{\omega}):=J_t(\omega'):=\omega'(t) \textrm{ on } \bigl[0,
T_{\infty}(\omega')\bigr)
\]
the canonical jump process, and let $(\cH_t)_{t\geq0}$ be its natural filtration, i.e., $\mathcal{H}_t = \sigma(J_s\,|\,s \leq t)$, with $\mathcal{H} = \mathcal{H}_{\infty}$. Note that $\{T_n\}_{n\in\mathbb{N}}$ are $(\mathcal{H}_t)$- and $(\mathcal{G}_t)$-stopping times if interpreted as $T_n(\widetilde{\omega})=T_n(\omega')$.
\end{itemize}
It is also useful to introduce the larger filtration $(\overline{\cG}_t)_{t\geq0}$ defined by $\overline{\cG}_t:=\cF_{\infty}\otimes\cH_t$, for all $t\geq0$. In particular, observe that $\cG_t\subseteq\overline{\cG}_t$, for all $t\geq0$, and $\overline{\cG}_0=\cF_{\infty}\otimes\{\emptyset,\Omega'\}$. The proofs of the following results are rather technical and, hence, postponed to Appendix~\ref{app:proof}.

\begin{proposition}	\label{prop:existenceY}
Suppose that Assumption \ref{ass:ex_sol_moment} holds and let $(\widetilde{\Omega}, \mathcal{G}, (\mathcal{G}_t)_{t\geq0})$ and $(\overline{\cG}_t)_{t\geq0}$ be defined as above. 
Then there exists a probability measure $\widetilde{\mathbb{Q}}$ on $(\widetilde{\Omega}, \mathcal{G})$ satisfying $\widetilde{\mathbb{Q}}|_{\mathcal{F}}=\mathbb{Q}$ and a c\`adl\`ag $(\cG_t)$-adapted $\mathbb{R}_+$-valued finite activity pure jump process $Y^{\perp}$  with jump measure $K_t\bigl(\omega,Y_{t-}^{\perp}(\omega,\omega'), d\xi\bigr)dt$ with respect to both filtrations $(\cG_t)_{t\geq0}$ and $(\overline{\cG}_t)_{t\geq0}$ and $K_t(\omega,y,d\xi)$ satisfies \eqref{eq:momentproblem}-\eqref{eq:mart_cond}.
\end{proposition}

The next lemma shows that the semimartingale property as well as the semimartingale characteristics of $(X,\Yhat)$ are not altered when considered on the extended filtered probability space $(\Omega,\cG,(\cG_t)_{t\geq0},\QQtilde)$. Moreover, besides satisfying the consistency condition \eqref{eq:cons_cond_Yperp}, the process $Y^{\perp}$ is locally independent of $(X,\Yhat)$ in the extended filtered probability space.  

\begin{lemma}	\label{lemma:immersion}
Suppose that Assumption \ref{ass:ex_sol_moment} holds and let the process $Y^{\perp}$ be constructed as in Proposition \ref{prop:existenceY}. Then, on the extended filtered probability space $(\Omega,\cG,(\cG_t)_{t\geq0},\QQtilde)$, the following hold:
\begin{enumerate}
\item $(X,\Yhat)$ is an $\RR^{d+1}$-valued semimartingale with the same  characteristics as in the original filtered probability space $(\Omega,\cF,(\cF_t)_{t\geq0},\QQ)$;
\item $Y^{\perp}$ is locally independent of $(X,\Yhat)$. 
Moreover, the process $\bigl(\exp(u_iY^{\perp}_t-\int_0^t\Psi^{Y^{\perp}}_s(u_i)ds)\bigr)_{t\geq0}$ is a $\bigl(\QQtilde,(\overline{\cG}_t)_{t\geq0}\bigr)$-martingale as well as a $\bigl(\QQtilde,(\cG_t)_{t\geq0}\bigr)$-martingale, for all $i=1,\ldots,m$.
\end{enumerate} 
\end{lemma}

We are now in a position to prove that, on the extended filtered probability space $(\Omega,\cG,(\cG_t)_{t\geq0},\QQtilde)$, step (c) of the model construction procedure described at the beginning of Section~\ref{sec:construction} can be successfully achieved and, hence, the three requirements of part (iii) of Theorem \ref{th:multicurve} are satisfied.

\begin{theorem}	\label{thm:constr_final}
Suppose that Assumption \ref{ass:ex_sol_moment} holds and let the process $Y^{\perp}$ be constructed as in Proposition \ref{prop:existenceY}.
Then, on the extended filtered probability space $(\Omega,\cG,(\cG_t)_{t\geq0},\QQtilde)$, the process $Y:=\Yhat+Y^{\perp}$ is compatible with the building blocks $(X,\Yhat,u_1,\ldots,u_m,f_0,\eta^1_0,\ldots,\eta^m_0,\widetilde{\sigma},\sigma^1,\ldots,\sigma^m)$, in the sense of Definition~\ref{def:compatible}.
\end{theorem}

\begin{remark}	\label{rem:driftdiff}
We want to point out that the present construction can be rather easily extended in order to include non-null drift and diffusion components in the process $Y^{\perp}$, by adapting the proof of Proposition \ref{prop:existenceY} in the spirit of \cite[Theorem A.4]{cfmt11}, under suitable hypotheses on the diffusion component so that the local independence as well as the martingale property are ensured. In a similar way, the requirement that the process $Y^{\perp}$ be of finite activity can also be easily relaxed.
\end{remark}

\section{Model implementation, calibration and tractable specifications}
\label{sec:implementation}

This section is devoted to several aspects related to the practical implementation of HJM-type multiple yield curve models. We start in Section~\ref{subsec:calibr} by  providing some general guidelines to model implementation and calibration. In Section~\ref{sec:gen_pricing} we present model-free valuation formulas for typical interest rate derivatives, while in Section~\ref{affinespecification} we introduce a tractable specification based on affine processes.
Further considerations on the practical applicability of our general HJM framework are given in Section~\ref{sec:relations}, where we shall discuss the relations with specific multiple curve models proposed in the literature, as well as in~\cite{CFGaffine}.

\subsection{General aspects of model implementation and calibration}	\label{subsec:calibr}
We now give some general guidelines for the implementation of an HJM-type multiple yield curve model, assuming that we can observe the following data:
\begin{enumerate}
\item market quotes for linear interest rate derivatives such as overnight indexed swaps, interest rate swaps and basis swaps (see Section~\ref{sec:noopt} below);\label{linearquotes}
\item market quotes for non-linear vanilla European interest rate derivatives such as caps/floors or swaptions (see Section~\ref{sec:option} below). 
\end{enumerate}

Given a set of market quotes for linear interest rate derivatives from \ref{linearquotes}, the preliminary step towards model calibration  consists in bootstrapping the term structures of OIS bonds and LIBOR rates implied by market data (by proceeding e.g. along the lines of~\cite{fries12}). The output of this first step are the initially observed term structures $\{B^M(0,T), \ T\geq 0\}$ and $\{L^M_0(T,T+\delta), \ T\geq 0,\delta\in \cD\}$, from which the spreads $\{S^{\delta,M}(0,T), \ T\geq 0,\delta\in\cD\}$ can be directly obtained (the superscript $M$ is meant to emphasize  the fact that these quantities are implied by market data).
Let us remark that this first part of the calibration procedure is obviously model independent, since it only relies on the general pricing formulas for linear interest rate derivatives given in Section~\ref{sec:noopt} below.

The next step in the implementation of an HJM-type multiple yield curve model requires of course the specification of a concrete model. In particular, in line with the standard HJM modeling approach, the driving semimartingales $(X,Y)$ and the volatilities $\widetilde{\sigma},\sigma^1,\ldots,\sigma^m$ have to be carefully chosen in order to ensure a satisfactory analytical tractability as well as the requirements for the multiplicative spreads introduced in Section~\ref{subsec:model_spread}. 
To this effect, a general model construction procedure has been provided in Section~\ref{sec:construction}, where a risk-neutral HJM-type multiple yield curve model is constructed starting from a given set of building blocks (see Definition~\ref{def:BB}). More specifically, this requires the specification of the driving semimartingale $X$, of the process $Y^{\parallel}$ (the dependent part of $Y$ with respect to $X$) as well as of the volatility processes $\widetilde{\sigma},\sigma^1,\ldots,\sigma^m$.
In particular, if one chooses a time-inhomogeneous L\'evy process as driving process (as considered e.g. in Section~\ref{sec:HJMmodels}), then the model can be implemented by relying on standard implementation and calibration techniques for L\'evy-driven HJM models which are well documented in the literature (see e.g.~\cite[Section  2.6]{phdkluge},~\cite[Chapter 4]{phdkoval},~\cite{EbKov2006},~\cite{ek07} as well as~\cite{DT13} in the case of general  HJM models driven by Brownian motion). We also refer the reader to~\cite[Section 4]{CGNS:13} for a detailed calibration with respect to swaption data of a L\'evy-driven multiple curve HJM model with Vasi\v cek-type volatility functions (see Section~\ref{sec:HJMmodels} for a rapid overview of this model specification).

However, multiple curve extensions of classical HJM specifications and calibration techniques do not represent the only feasible approaches for the practical implementation of an HJM-type multiple yield curve model. Indeed, as will be shown  in Section~\ref{affinespecification} below (and referring to the companion paper~\cite{CFGaffine} for full details), a model specification based on affine processes represents an especially flexible and tractable solution. Besides being able to automatically fit the term structures calibrated from market data (see e.g. Definition~\ref{def:deterministiShift}), an affine specification can easily ensure that spreads are greater than one and ordered with respect to the tenor's length. Moreover, due to the intrinsic analytical tractability of affine processes, Fourier techniques lead to efficient pricing formulas/approximations for caps and swaptions.


In the affine case, the existence of efficient pricing formulas allows to 
 proceed with a standard calibration procedure, consisting in a search for a parameter vector that minimizes the distance between market implied and model implied volatilities. The calibration may be performed over cap or swaption market quotes (currently, the joint calibration to both types of derivatives seems to be an open challenge). In the former case, since the market only provides cap and not caplet volatilities, the calibration routine can be simplified by first bootstrapping a surface of caplet volatilities from the observed surface of cap volatilities, along the lines of~\cite{bloomberg12}. Finally, as we show in~\cite{CFGaffine}, the multiple curve extension of affine Libor models proposed by~\cite{GPSS14} can be regarded as a discrete tenor version of our affine specification, so that the calibration techniques (with respect to caplet data) introduced in that paper can be also employed in our context. 
 
\subsection{General pricing formulae}	\label{sec:gen_pricing}
The quantity $S^{\delta}(t,T)$ plays a pivotal role in the valuation of interest rate products. We here derive clean valuation formulas in the spirit of Appendix~\ref{sec:FRArates} assuming perfect collateralization and a collateral rate equal to the OIS rate.

\subsubsection{Linear products}\label{sec:noopt}
The prices of linear interest rate products (i.e., without optionality features) can be directly expressed in terms of the basic quantities $B(t,T)$ and $S^{\delta}(t,T)$.

\emph{Forward rate agreement.}
A forward rate agreement (FRA) starting at $T$, with maturity $T+\delta$, fixed rate $K$ and notional $N$ is a contract which pays at time $T+\delta$ the following amount
\begin{align*}
\Pi^{FRA}(T+\delta; T,T+\delta,K,N)=N\delta\bigl(L_T(T,T+\delta)-K\bigr).
\end{align*}
The value of such a claim at time $t\leq T$ is
\begin{equation}\label{eq:FRA}
\begin{split}
\Pi^{FRA}(t; T,T+\delta,K,N)&=NB(t,T+\delta)\delta\,\Excond{\QQ^{T+\delta}}{L_T(T,T+\delta)-K}{\cF_t}\\
&=N\bigl(B(t,T)S^{\delta}(t,T)-B(t,T+\delta)(1+\delta K)\bigr).
\end{split}
\end{equation}

\emph{Overnight indexed swap.}
An overnight indexed swap (OIS) is a contract where two counterparties exchange two streams of payments: the first one is computed with respect to a fixed rate $K$, whereas the second one is indexed by an overnight rate (EONIA). Let us denote by $T_1, \ldots, T_n$ the payment dates, with $T_{i+1}-T_i=\delta$ for all $i=1,\ldots,n-1$. The swap is initiated at time $T_0\in[0,T_1)$.  
The value of the OIS at time $t\leq T_0$, with notional $N$, can be expressed as (see e.g.~\cite[Section 2.5]{fitr12})
\begin{align*}
\Pi^{OIS}(t; T_{1},T_{n},K,N)=N\left(B(t,T_0)-B(t,T_n)-K\delta\sum_{i=1}^{n} B(t,T_i)\right).
\end{align*}
Therefore, the OIS rate $K^{OIS}$, which is by definition the value for $K$ such that the OIS contract has zero value at inception, is given by
\begin{align*}
K^{OIS}(T_1,T_n)=\frac{B(t,T_0)-B(t,T_n)}{\delta\sum_{k=1}^{n} B(t,T_k)}.
\end{align*}

\emph{Interest rate swap.}
In an interest rate swap (IRS), two streams of payments are exchanged between two counterparties: the first cash flow is computed with respect to a fixed rate $K$, whereas the second one is indexed by the prevailing Libor rate. 
The value of the IRS at time $t\leq T_0$, where $T_0$ denotes the inception time, is given by
\begin{align}
\Pi^{IRS}(t; T_{1},T_{n},K,N)
&=N\sum_{i=1}^{n}\left(B(t,T_{i-1})S^{\delta}(t,T_{i-1})-B(t,T_{i})(1+\delta K)\right).	\label{eq:IRS}
\end{align}
The swap rate $K^{IRS}$, which is by definition the value for $K$ such that the contract has zero value at inception, is given by
\begin{align*}
K^{IRS}(T_{1},T_{n},\delta)
= \frac{\sum_{i=1}^n\bigl(B(t,T_{i-1})S^{\delta}(t,T_{i-1})-B(t,T_i)\bigr)}{\delta\sum_{i=1}^nB(t,T_i)}
= \frac{\sum_{i=1}^nB(t,T_i)L_t(T_{i-1},T_i)}{\sum_{i=1}^nB(t,T_i)}.
\end{align*}

\emph{Basis swap.}
A basis swap is a special type of interest rate swap where two cash flows related to Libor rates associated to different tenors are exchanged between two counterparties. For instance, a typical basis swap may involve the exchange of the 3M against the 6M Libor rate. Following the standard conventions in the Euro market (see~\cite{AB:13}), the basis swap is equivalent to a long/short position on two different IRS which share the same fixed leg. Let  $\cT^1=\left\{T^1_0,\cdots, T^1_{n_{1}}\right\}$, $\cT^2=\left\{T^2_0,\cdots T^2_{n_{2}}\right\}$ and $\cT^3=\left\{T^3_0,\cdots, T^3_{n_{3}}\right\}$, with $T^1_{n_{1}}=T^2_{n_{2}}=T^2_{n_{3}}$, $\cT^1\subset \cT^2$, $n_1<n_2$ and corresponding tenor lengths $\delta_1 > \delta_2$ (and arbitrary $\delta_3$). The first two tenor structures on the one side and the third on the other are associated to the two floating and to the single fixed leg, respectively.  As usual, we denote by $N$ the notional of the swap (initiated at time $T^1_0=T^2_0=T^3_0$). The value at time $t\leq T^1_0$ is given by 
\begin{align*}
\Pi^{BSW}(t;\cT^1,\cT^2,\cT^3,N)
&=
N\left(\sum_{i=1}^{n_1}\bigl(B(t,T^1_{i-1})S^{\delta^1}(t,T^1_{i-1})-B(t,T^1_i)\bigr)\right.\\
&\quad-\sum_{j=1}^{n_2}\bigl(B(t,T^2_{j-1})S^{\delta^2}(t,T^2_{j-1})-B(t,T^2_j)\bigr)-\left.K\sum_{\ell=1}^{n_3}\delta^3B(t,T^3_{\ell})\right).
\end{align*}
The value $K^{BSW}$ (called \textit{basis swap spread}) such that the value of the contract at initiation is zero is then given by
\begin{align*}
K^{BSW}(\cT^1,\cT^2,\cT^3)&=
\frac{\sum_{i=1}^{n_1}\bigl(B(t,T^1_{i-1})S^{\delta^1}(t,T^1_{i-1})-B(t,T^1_i)\bigr)-\sum_{j=1}^{n_2}\bigl(B(t,T^2_{j-1})S^{\delta^2}(t,T^2_{j-1})-B(t,T^2_j)\bigr)}{\delta_3\sum_{\ell=1}^{n_3}B(t,T^3_{\ell})}.
\end{align*}
It is interesting to observe that, prior to the financial crisis, the value of $K^{BSW}$ used to be (approximately) zero.

\subsubsection{Products with optionality features}	\label{sec:option}
In this section,  we report general valuation formulas  for plain vanilla interest rate products such as European caplets and swaptions. 

\emph{Caplet.}
The price at time $t$ of a caplet with strike price $K$, maturity $T$, settled in arrears at $T+\delta$, is given by
\begin{align*}
\Pi^{CPLT}(t; T,T+\delta,K,N)&=NB_t\delta\,\Excond{}{\frac{1}{B_{T+\delta}}\Bigl(L_T(T,T+\delta)-K\Bigr)^+}{\cF_t}.
\end{align*}
This valuation formula above admits the following representation in terms of $S^{\delta}(t,T)$.
\begin{align}
\Pi^{CPLT}(t; T,T+\delta,K,N)
&=N\Excond{}{\frac{B_t}{B_T}\Bigl(S^{\delta}(T,T)-(1+\delta K)B(T,T+\delta)\Bigr)^+}{\cF_t}	\label{eq:caplet_gen}
\end{align}
\begin{remark}
Note that the valuation formula \eqref{eq:caplet_gen} in the classical single curve setting (i.e., under the assumption that $S^{\delta}(T,T)$ is identically equal to one), reduces to the classical relationship between a caplet and a put option on a zero-coupon bond with strike $1/(1+\delta K)$. 
\end{remark}

\emph{Swaption.}
We consider a standard European payer swaption with maturity $T$, written on a (payer) interest rate swap starting at $T_0=T$ and payment dates $T_1,..., T_n$, with $T_{i+1}-T_i=\delta$ for all $i=1,\ldots,n-1$, with notional $N$. Due to formula~\eqref{eq:IRS}, the value of such a claim at time $t$ is given by
\[
\begin{aligned}
\Pi^{SWPTN}(t; T_{1},T_{n},K,N)
&=N\mathbb{E}^{}\left[\frac{B_t}{B_T}\left(\sum_{i=1}^{n}B(T,T_{i-1})S^{\delta}(T,T_{i-1})-(1+\delta K)B(T,T_i)\right)^+\biggr|\cF_t\right].
\end{aligned}
\]

%

\subsection{Models based on affine processes}\label{affinespecification}

In this section, we propose a flexible and tractable specification of the general framework of Section~\ref{sec:model_spread} based on affine processes. Concerning the analytical properties and the full characterization of such Markov processes on different state spaces we refer to 
\cite{dfs03}, \cite{cfmt11}, and \cite{krm12}.
Throughout this section $\mathbb{T}$ denotes a fixed time horizon and we work on a filtered probability space $(\Omega,\cF,(\cF_t)_{0\leq t\leq \mathbb{T}},\QQ)$, where $\QQ$ is a risk neutral measure under which OIS bonds and FRA contracts when discounted with the OIS bank account are martingales.
In order to define \emph{affine multiple yield curve models} we consider a stochastic process $\cX=\left(X_t,Y_t,Z_t\right)_{0\leq t\leq \mathbb{T}}$ whose state space is denoted by $D=D_X\times \RR^{n+1}$,  meaning that $X$ takes values on $D_X$, which is assumed to be a closed convex subset of a real Euclidean vector space $V$ with scalar product $\langle \cdot, \cdot\rangle$ and $(Y,Z)$ is $\RR^{n+1}$-valued.  The process $X$ will represent the general driving process, the process $Y$ is related to the log-spot spread
exactly as in Section~\ref{sec:model_spread} and $Z$ to the OIS bank account via $B=\exp(-Z)$. In order to qualify for an affine multiple yield curve model the process $\cX=(X,Y,Z)$ has to satisfy the ``affine property'' in the following sense:
\begin{itemize}
\item[A1)] $\cX$ is a stochastically continuous time-homogeneous Markov process with state space $D$;
\item[A2)] the Fourier-Laplace transform of $\cX_t=(X_t,Y_t,Z_t)$ has exponentially affine dependence on the initial states $(x,y,z)$, that is, there exists function $(t,v,u,w)\mapsto \phi(t,v,u,w)$ and $(t,v,u,w)\mapsto \psi(t,v,u,w)$ such that
\begin{align*}
\mathbb{E}\left[e^{\langle v,X_t\rangle+u^{\top}Y_t+wZ_t}\right]=e^{\phi(t,v,u,w)+\langle \psi(t,v,u,w),x\rangle+u^{\top}y+ wz},
\end{align*}
for all $(x,y,z) \in D$, $(v,u,w)\in\mfU$ and $t\in[0,\mathbb{T}]$, where the set $\mfU$ is defined by
\[
\mfU:=\bigl\{\left.\zeta\in (V+\im V)\times \CC^{n+1}\right| \mathbb{E}\bigl[e^{\langle \zeta,\cX_t\rangle}\bigr]<\infty, \, \forall t \in [0,\mathbb{T}]\bigr\}.
\]
\end{itemize}

\begin{remark}
Notice that the above form of the Fourier-Laplace transform implies that the characteristics of $(Y,Z)$ only depend on $X$, whence we do not consider all possible affine processes on $D$.
\end{remark}

We now introduce affine multiple yield curve models via the following definition (see also~\cite{CFGaffine}).

\begin{definition}	\label{def:aff_model}
An \emph{affine multiple yield curve model} is defined via
\begin{enumerate}
\item a process $\cX=\left(X,Y,Z\right)$ satisfying A1) and A2), with $X$ and $(Y,Z)$ taking values in $D_X$ and in $\RR^{n+1}$, respectively, and vectors $u_0=0,u_1,\ldots,u_m\in\RR^n$ satisfying $(0,u_i,1)\in\mfU$, for all $i=0,1,\ldots,m$;
\item a (risk-free) OIS bank account given by $B_t=e^{-Z_t}=e^{\int_0^tr_sds}$, for all $t\geq0$, where
the OIS short rate (determining the process $Z$) satisfies
\[
r_t=l+\langle \lambda,X_t\rangle, \quad t \in [0, \mathbb{T}],
\]
for some $l\in\RR$ and $\lambda\in\RR^n$;
\item a family of (multiplicative) spot spreads $\bigl\{\bigl(S^{\delta_i}(t,t)\bigr)_{t\in[0,\mathbb{T}]},i \in \{1, \ldots, m\}\bigr\}$ modeled as 
\[
S^{\delta_i}(t,t)=e^{u_i^{\top} Y_t},\quad t \in [0, \mathbb{T}], \quad i \in \{1, \ldots, m\},
\]
for the vectors $u_1,\ldots,u_m\in\RR^n$.
\end{enumerate}
\end{definition}

\begin{remark}
Affine multiple yield curve models represent the natural extension of classical affine short rate models to the multi-curve setting, since the latter are usually specified through (i) and (ii) in the above definition.
\end{remark}

\begin{proposition}
In an affine multiple yield curve model the OIS bonds $\bigl\{\bigl(B(t,T)\bigr)_{t\in[0,T]},T\leq\mathbb{T}\bigr\}$
and the multiplicative spreads $\bigl\{\bigl(S^{\delta_i}(t,T)\bigr)_{t\in[0,T]},T\leq\mathbb{T},i\in \{1, \ldots,m\}\bigr\}$ are exponentially affine. More precisely, we have
\begin{equation}	\label{bondaffinemodel}
B(t,T)=\mathbb{E}\left[B_t/B_T|\mathcal{F}_t\right]=\mathbb{E}\left[e^{Z_T-Z_t}|\mathcal{F}_t\right]=e^{\phi(T-t,0,0,1)+\langle\psi(T-t,0,0,1),X_t\rangle},
\end{equation}
and
\begin{equation}	\label{spreadaffinemodel}
S^{\delta_i}(t,T)=\frac{\Excond{}{e^{Z_T+u_i^\top Y_T}}{\cF_t}}{\Excond{}{e^{Z_T}}{\cF_t}}=e^{u_i^\top Y_t+\phi(T-t,0,u_i,1)-\phi(T-t,0,0,1)+\langle \psi(T-t,0,u_i,1)-\psi(T-t,0,0,1),X_t\rangle},
\end{equation}
where $\phi$ and $\psi$ denote the characteristic exponents of the affine process $\mathcal{X}=(X,Y,Z)$ as given in A2).
\end{proposition}

\begin{proof}
The form of the bond prices is simply a consequence of the affine property of $\mathcal{X}$ (under the risk neutral measure $\mathbb{Q}$).
Concerning the form of the spreads it suffices to note that 
\begin{align*}
S^{\delta_i}(t,T)&=\mathbb{E}^{\mathbb{Q}^T}[S^{\delta_i}(T,T)| \mathcal{F}_t]=\mathbb{E}^{\mathbb{Q}^T}[e^{u_i^{\top} Y_T}|\mathcal{F}_t]
=\frac{\mathbb{E}^{\mathbb{Q}}[e^{u_i^{\top} Y_T}/B_T|\mathcal{F}_t]}{\frac{B(t,T)}{B_t}}=\frac{\Excond{}{e^{Z_T+u_i^\top Y_T}}{\cF_t}}{\Excond{}{e^{Z_T}}{\cF_t}}\,.
\end{align*}
The right hand side of~\eqref{spreadaffinemodel} then follows again from the affine property.
\end{proof}

Note that if $Y$ lies in some cone $C \subset \mathbb{R}^n$ and $u_1, \ldots u_m$ satisfy the requirements of Corollary~\ref{cor:ordered}, the multiplicative spreads are greater than one and ordered with respect to the tenor's length. Besides the modeling flexibility and tractability ensured by affine processes, this represents one of the main advantages of the spread specification~\eqref{spreadaffinemodel}.

\begin{remark}
Let us briefly comment on possible specifications of the process $(X,Y,Z)$, in particular with a view to obeying the order relations.
One choice is to take $X$ with values in a cone,  for instance $\mathbb{R}^d_+$ or  the set of positive semidefinite matrices  $S_d^+$, and to specify $Y$  via  $Y=g(X)+L$, where $g$ is an affine function taking values in $\mathbb{R}^n_+$ and $L$ denotes an $\mathbb{R}^n_+$-valued L\'evy process. One concrete simple specification with only one spread is to specify $X$ as $S_2^+$-valued Wishart process (see~\cite{bru91}), set and $Y=X_{11}$ and $Z=-\int_0^{\cdot} X_{22,s} ds$. 
In line with Remark~\ref{rem:short_spread} and in analogy to the bank account, $Y$ can also be specified as $Y=\int_0^{\cdot} q(X_s) ds$, where $q: D_X \to \mathbb{R}^n$ here denotes an affine function.
Concrete model specifications and pricing of interest derivatives in the above introduced affine multiple yield curve setting are considered in the follow-up paper~\cite{CFGaffine}.
\end{remark}

The above definition of affine multiple yield curve model can be directly mapped into the general setup of HJM-type multiple yield curve models from Section \ref{sec:model_spread} via the following proposition, which also shows that the risk neutral property (see Definition~\ref{def:riskneutral}) is satisfied by construction.

\begin{proposition}
Every affine multiple yield curve model is a risk neutral HJM-type multiple yield curve model where
\begin{enumerate}
\item the driving process is $X$;
\item the bank account is given by $B_t=e^{-Z_t}$, for all $t\geq0$;
\item the log-spot spread is given by $\log S^{\delta_i}(t,t)=u_i^\top Y_t$, for all $i=1,\ldots,m$ and $t\geq0$;
\item the forward rate $f_t(T)$ and the forward spread rates $\eta^i_t(T)$ are given by
\begin{align}
f_t(T)&=-F\bigl(\psi(T-t,0,0,1),0,1\bigr)-\bigl\langle R\bigl(\psi(T-t,0,0,1),0,1\bigr),X_t\bigr\rangle,\\
\begin{split}
\eta^i_t(T)&=F\bigl(\psi(T-t,0,u_i,1),u_i,1\bigr)-F\bigl(\psi(T-t,0,0,1),0,1\bigr)\\
&\quad+\bigl\langle R\bigl(\psi(T-t,0,u_i,1),u_i,1\bigr)-R\bigl(\psi(T-t,0,0,1),0,1\bigr),X_t\bigr\rangle,
\end{split}
\end{align}
for all $t\leq T$, $T\geq0$ and $i=1,\ldots,m$, where $F(\psi(t,v,u,w),u,w)=\partial_t \phi(t,v,u,w)$ and $R(\psi(t,v,u,w),u,w)=\partial_t\psi(t,v,u,w)$.
\end{enumerate}
\end{proposition}
\begin{proof}
The first three claims follow upon direct inspection of the definition. The expressions for $f_t(T)$ and $\eta^i_t(T)$ are obtained from \eqref{bondaffinemodel} and \eqref{spreadaffinemodel} by simply noting that
\begin{align*}
-\int_t^Tf_t(s)ds&=\phi(T-t,0,0,1)+\langle\psi(T-t,0,0,1),X_t\rangle,\\
\int_t^T\eta_t(s)ds&=\phi(T-t,0,u_i,1)-\phi(T-t,0,0,1)+\langle \psi(T-t,0,u_i,1)-\psi(T-t,0,0,1),X_t\rangle
\end{align*}
and by differentiating both sides, which is possible due to regularity of affine processes (see~\cite{CT:13}). The risk neutral property follows from the fact that $\mathbb{Q}$ is a risk neutral measure. 
\end{proof}


\subsubsection*{A deterministic shift extension} 	\label{sec:shift}
In view of practical implementations, a relevant issue is represented by the capability of the model to provide an exact fit to the initially observed term structures $\{B^M(0,T), T\leq\mathbb{T}\}$ and $\{S^{\delta_i, M}(0,T), T\leq\mathbb{T},\delta_i\in\mathcal{D}\}$ of risk-free bonds and spreads (at date $t=0$). 
In the spirit of \cite{bm01}, we can extend our affine specification in order to ensure an exact fit to the initially observed term structures via the following definition. 

\begin{definition}\label{def:deterministiShift}
A \emph{shifted affine multiple yield curve model} is defined by introducing the following specifications for OIS bonds and multiplicative spreads (compare with Definition~\ref{def:aff_model}):
\begin{align*}
B(t,T)&=\frac{B^M(0,T)}{B^M(0,t)}\frac{e^{\phi(t,0,0,1)+\langle\psi(t,0,0,1),X_0\rangle}}{e^{\phi(T,0,0,1)+\langle\psi(T,0,0,1),X_0\rangle}}e^{\phi(T-t,0,0,1)+\langle\psi(T-t,0,0,1),X_t\rangle}\\
S^{\delta_i}(t,T)&=S^{\delta_i,M}(0,T)\frac{e^{u_i^\top Y_t+\phi(T-t,0,u_i,1)-\phi(T-t,0,0,1)+\langle \psi(T-t,0,u_i,1)-\psi(T-t,0,0,1),X_t\rangle}}{e^{u_i^\top Y_0+\phi(T,0,u_i,1)-\phi(T,0,0,1)+\langle \psi(T,0,u_i,1)-\psi(T,0,0,1),X_0\rangle}}.
\end{align*}
\end{definition}
With the above specification, the term structures obtained by bootstrapping market data can be regarded as inputs of the model, while model parameters are calibrated to volatility surfaces of derivatives like caplets or swaptions. Let us remark that with the above shift extension it is not a priori guaranteed that $S^\delta(t,T)>1$. 
This generalization is related to the Hull-White extension for affine models in the sense that the state independent characteristics of the underlying affine process become time dependent deterministic functions. 
Deterministic shift extensions have been recently employed also in \cite{GM:14}.

\section{Relations with other multiple yield curve modeling approaches}
\label{sec:relations}

In the present section, we briefly discuss how our HJM-type framework relates to several multiple yield curve models that have been recently proposed in the literature. In particular, it will be shown that most of the existing modeling approaches can be recovered from our general setting. For consistency of exposition, we shall adapt to our notation the original notation used in the papers mentioned below. 
For a more detailed study of the relations between our general HJM-type framework and multiple curve models based on affine process we refer to~\cite{CFGaffine}.

\subsection{HJM models}	\label{sec:HJMmodels}

The multiple curve HJM models proposed in \cite{CGNS:13,cre12,mopa10} can be rather easily recovered from our general framework. In particular, similarly as in our approach, \cite{CGNS:13} directly model FRA rates, while the risk-free term structure is modeled as in the classical HJM setup. Discounted prices of risk-free bonds are specified as $B(t,T)/B_t=B(0,T)\exp\bigl(-\int_0^t\widetilde{A}(s,T)ds-\int_0^t\widetilde{\Sigma}(s,T)dX_s\bigr)$, for all $t \leq T$ and $T\geq0$, where $\widetilde{A}$ and $\widetilde{\Sigma}$ are deterministic functions and $X$ is a multivariate L\'evy process. 
The martingale property of $B(\cdot,T)/B$, for every $T\geq0$, is ensured by the classical drift condition $\widetilde{A}(t,T)=\Psi^X\bigl(-\widetilde{\Sigma}(t,T)\bigr)$, for all $0\leq t \leq T$, with $\Psi^X$ denoting the L\'evy exponent of $X$, and the function $\widetilde{\Sigma}$ is assumed to be uniformly bounded.
FRA rates are also modeled via an HJM approach which, up to a deterministic shift, corresponds to the following specification of $S^{\delta}(t,T)$: 
\[	\begin{aligned}
S^{\delta}(t,T)
&= S^{\delta}(0,T)\exp\left(-\int_0^t\Psi^{T,X}\bigl(\Sigma^{\delta}(s,T)\bigr)ds
+\int_0^t\Sigma^{\delta}(s,T)dX_s\right),
\end{aligned}	\]
where $\Sigma^{\delta}(t,T):=\varsigma(t,T,T+\delta)-\widetilde{\Sigma}(t,T+\delta)+\widetilde{\Sigma}(t,T)$, according to the notation of~\cite{CGNS:13}, and $\Psi^{T,X}$ denotes the local exponent of $X$ under the $T$-forward measure $\mathbb{Q}^T$ and where we have used the drift condition (12) of \cite{CGNS:13}. In particular, by means of straightforward computations, one can then obtain a representation of the form $S^{\delta}(t,T)=\exp\bigl(Z^{\delta}_t+\int_t^T\eta^{\delta}_t(s)ds\bigr)$, for all $t \leq T$ and $T\geq0$.
By relying on analogous considerations, it can be shown that the multi-curve HJM models of \cite{cre12} and \cite{mopa10} can be also recovered from our framework\footnote{In particular, in the paper \cite{cre12} (according to the notation used therein), it holds that $Z^{\delta}_t=\int_t^{t+\delta}g_t(s)\,ds$, for any tenor $\delta>0$, with $g_t(T)$ representing the spread between risky and risk-free instantaneous $T$-forward rates.}.

\subsection{Short rate models}

As mentioned in the introduction, models based on short rates have also been proposed for modeling multiple curves, see in particular \cite{ken10,kitawo09,MR14}. This modeling approach can also be embedded within our general framework. In particular, \cite{kitawo09} formulate a simple short rate model that allows for the consistent pricing of fixed income products related to different curves. For simplicity, let us consider the case of two interest rate curves: the discounting curve, denoted by $D$, and the Libor curve, denoted by $L$. To the two curves $D$ and $L$, \cite{kitawo09} associate the short rate processes $(r^D_t)_{t\geq0}$ and $(r^L_t)_{t\geq0}$, with corresponding bond prices
\begin{equation}	\label{KTW-1}
B(t,T) := \EE^{\QQ}\left[e^{-\int_t^Tr^D_udu}\bigr|\cF_t\right]
\qquad\text{and}\qquad
B^{\delta}(t,T) := \EE^{\QQ_L}\left[e^{-\int_t^Tr^L_udu}\bigr|\cF_t\right],
\end{equation}
for all $t\leq T$ and $T\geq0$, where the measure $\QQ_L\sim\QQ$ represents a risk neutral measure with respect to the ``$L$ savings account'' (see \cite{kitawo09}, Section 2).
The Libor rate with tenor $\delta$ is then given by $L_T(T,T+\delta)=\left(1/B^{\delta}(T,T+\delta)-1\right)/\delta$, for all $T\geq0$. Hence, according to our notation\footnote{Note that our multiplicative spread $S^{\delta}(t,T)$ corresponds to $K_L(t,T,T+\delta)$, according to the notation adopted in \cite{kitawo09}.}, 
\begin{equation}	\label{KTW-2}	
S^{\delta}(t,T) 
= \EE^{\QQ^{T+\delta}}\bigl[1+\delta L_T(T,T+\delta)|\cF_t\bigr]\frac{B(t,T+\delta)}{B(t,T)}
= \frac{\EE^{\QQ}\left[e^{-\int_0^Tr^D_udu}\frac{1}{H^L(T,T+\delta)}\bigr|\cF_t\right]}
{\EE^{\QQ}\left[e^{-\int_0^Tr^D_udu}|\cF_t\right]}
= \frac{\EE^{\QQ}\left[e^{Z_T+Y_T}|\cF_t\right]}
{\EE^{\QQ}\left[e^{Z_T}|\cF_t\right]},
\end{equation}
where $H^L(t,T):=B^{\delta}(t,T)/B(t,T)$, $Z_t:=-\int_0^tr^D_udu$ and $Y_t:=-\log H^L(t,t+\delta)$. 
Moreover, it holds that $H^L(t,T)=\EE^{\QQ_L}\bigl[\exp\bigl(-\int_t^Th^L_udu\bigr)|\cF_t\bigr]$, where $(h^L_t)_{t\geq0}$ is an Ornstein-Uhlenbeck spread process. Note also the similarity between the rightmost term in \eqref{KTW-2} and the general affine specification \eqref{spreadaffinemodel} of multiplicative spreads.
Analogous considerations apply to the recent paper \cite{MR14}, where risk-free and risky bond prices are modeled as in \eqref{KTW-1}, but under a common risk-neutral measure $\mathbb{Q}$, and with a ``risky'' short rate obtained by adding to the risk-free rate $(r_t)_{t\geq0}$ a stochastic spread $(s_t)_{t\geq0}$. In particular, for a given tenor $\delta>0$ and for any $0\leq t\leq T$, our multiplicative spread $S^{\delta}(t,T)$ corresponds to the ratio $\bar{\nu}_{t,T}/\nu_{t,T}$, according to the notation of \cite{MR14}.
The short rate model proposed in \cite{ken10} is also rather similar, the main difference being that the risky short rate process is assumed to be specific to each tenor.  

\subsection{Lognormal Libor market models}

Similarly as in the original article~\cite{BGM:97}, we can also obtain a lognormal Libor market model for 
$L_t(T,T+\delta)$ within the above framework. Let $\delta$ be fixed and 
consider the setting of Remark~\ref{rem:short_spread} with $Y$ one-dimensional, given by $Y_t=\int_0^t q_s ds$, and $u=1$.
Assume furthermore that the driving process $X$ in our multiple yield curve model is a standard $d$-dimensional Brownian motion $W$ and suppose that the dynamics of $L_t(T,T+\delta)$ are given by 
\[
dL_t(T,T+\delta)= L_t(T,T+\delta)\beta_t(T)dW^{T+\delta}_t,
\]
where $\beta_t(T)$ is an $\mathbb{R}^d$-valued bounded deterministic function and $(W^{T+\delta}_t)_{t\geq0}$ denotes a $\mathbb{Q}^{T+\delta}$-Brownian motion.
Recalling that 
\[
1+\delta L_t(T,T+\delta)=S^{\delta}(t,T)\bigl(1+\delta L_t^D(T,T+\delta)\bigr)=e^{\int_0^t q_sds+\int_t^T \eta^{\delta}_t(u)du+\int_T^{T+\delta} f_t(u)du},
\]
applying It\^o's formula to both sides and comparing the diffusion coefficients, we obtain
\[
\Sigma_t(T)=\frac{\delta L_t(T,T+\delta)}{1+\delta L_t(T,T+\delta)}\beta_t(T)-\bigl(\widetilde{\Sigma}_t(T+\delta)-\widetilde{\Sigma}_t(T)\bigr).
\]
Supposing differentiability of  $T \mapsto \beta_t(T)$, we can derive an expression for $\sigma_t(T)$
\begin{align*}
\sigma_t(T)&=e^{-\int_0^t q_sds-\int_t^T \eta^{\delta}_t(u)du-\int_T^{T+\delta} f_t(u)du}\left((\eta^{\delta}_t(T)+f_t(T+\delta)-f_t(T))\beta_t(T)-\partial_T \beta_t(T)\right)+\partial_T \beta_t(T)\\
&\quad -\widetilde{\sigma}_t(T+\delta)+\widetilde{\sigma}_t(T).
\end{align*}
In order to study the existence of a solution to the S(P)DE for $\eta^{\delta}$ corresponding to this volatility structure, we can switch -- similarly as in Section~\ref{sec:exist_uniq} --
to the Musiela parametrization.
Under appropriate assumptions on the involved parameters $\beta$, $q$ and $\widetilde{\sigma}$, existence and uniqueness for $\eta^{\delta}$ can be obtained similarly as in Theorem~\ref{thm:existence_uniqueness}. 
This approach thus provides a theoretical justification in the multiple yield curve setting for the market practice to price caplets by means of Black's formula.

\appendix

\section{Pricing under collateral and FRA rates}		\label{sec:FRArates}


Let us here briefly review pricing under perfect collateralization for general derivatives which we then apply to the pricing of FRAs.
For a more detailed discussion on general valuation with collateralization and funding costs, we refer to the growing literature on this topic, e.g.,~\cite{BR:13} and the references therein. We here follow closely~\cite[Section 2.2]{fitr12}. 

Throughout let $(\Omega, \mathcal{F}, (\mathcal{F}_t)_{t\geq0}, \mathbb{P})$ be a filtered probability space, where $\PP$ stands for the statistical/historical probability measure. We consider OIS zero coupon bonds as
basic traded instruments, which play the role of risk-free zero coupon bonds in the classical setting. In order to guarantee no arbitrage we assume that:
\begin{enumerate}
\item there exists an OIS bank account denoted by $(B_t)_{t\geq 0}$ such that $B_t=\exp(\int_0^t r_s ds)$, where $r$ denotes the OIS short rate;
\item there exists an equivalent probability measure $\mathbb{Q}$ such that the OIS bonds for all maturities
are $\mathbb{Q}$-martingales when denominated in units of the OIS bank account.\footnote{Let us remark that in the presence of funding costs absence of arbitrage is implied by the existence of an equivalent measure under which the 
risky assets $S$ present in the market are (local) martingales when discounted with their corresponding funding rate $r^f$. This can be embedded in the classical framework where $\mathbb{Q}$ is a risk neutral measure with the risk-free (OIS) bank account as num\'eraire, by treating $e^{\int_0^{\cdot} -r_s^f+r_s}S$ as traded asset.}
\end{enumerate}
Let now $X$ be an $\mathcal{F}_T$-measurable payoff of some derivative security. We assume here a perfect collateral agreement where  100\% of the derivative's present value $V_t$ is posted in the collateral at any time $t < T$. 
The receiver of the collateral can invest it at risk-free rate $r$, corresponding to the OIS short rate, and has to pay an agreed \emph{collateral rate} $r^c_t$ to the poster of the collateral.
Applying risk neutral pricing, we obtain the following expression for the present value of the collateralized transaction
\[
V_t=\mathbb{E}^{\mathbb{Q}}\left[ e^{-\int_t^T r_s ds}X+\int_t^T e^{-\int_t^s r_u du} (r_s -r^c_s)V_s ds \, \Big |\, \mathcal{F}_t\right].
\]
As shown in~\cite[Appendix A]{fitr12}, this formula is equivalent to 
$
V_t=\mathbb{E}^{\mathbb{Q}}[ e^{-\int_t^T r^c_s ds}X \, |\,  \mathcal{F}_t].
$
Assuming that the collateral rate $r^c$ corresponds to the OIS short rate $r$, which is usually the case, we obtain the classical risk neutral valuation formula.

Since market quotes of FRAs correspond to perfectly collateralized contracts, where the \emph{collateral rate $r^c$ is assumed to be the OIS short rate} $r$, the above pricing approach is applied for the definition of FRA rates. 
As in classical interest rate theory, the FRA rate, denoted by $L_t(T,T+\delta)$, is the rate $K$ fixed at time $t$ such that the value of the FRA contract, whose payoff at time $T+\delta$ is given by $\delta(L_T(T,T+\delta)- K)$
 has value $0$. Therefore, it holds that, for all $t\in[0,T]$ and $T\geq0$,
\[
\mathbb{E}^{\mathbb{Q}}\left[e^{-\int_t^{T+\delta} r_s ds}(L_T(T,T+\delta)-K) \, \big| \, \mathcal{F}_t\right] \stackrel{!}{=}0.
\]
Hence by Bayes formula, 
\begin{align*}
L_t(T,T+\delta)=\mathbb{E}^{\mathbb{Q}^{T+\delta}}\left[L_T(T,T+\delta)\, \big|\, \mathcal{F}_t\right],
\end{align*}
where $\mathbb{Q}^{T+\delta}$ denotes the $(T+\delta)$-forward measure associated with the num\'eraire $B(\cdot,T+\delta)$ and 
density process $\frac{d\mathbb{Q}^{T+\delta}}{d\mathbb{Q}}|_{\mathcal{F}_t}=\frac{B(t,T+\delta)}{B_t B(0,T+\delta)}$. In particular, this provides a rigorous justification for the market practice of taking expression~\eqref{eq:defLibor} as the definition of fair FRA rates.

\section{Foreign exchange analogy}	\label{app:FX}

For simplicity of presentation, let us consider a fixed tenor $\delta$ and define artificial ``risky'' bond prices $B^{\delta}(t,T)$ at time $t$ and maturity $T$ for the tenor $\delta$ by the following relation, for all $t\leq T$ and $T\geq0$,
\[
L_t(T,T+\delta)=:\frac{1}{\delta} \left(\frac{B^{\delta}(t,T)}{B^{\delta}(t,T+\delta)}-1\right).
\]
While the family $\bigl\{\bigl(B(t,T)\bigr)_{t\in[0,T]},T\geq0\bigr\}$ represents prices of \emph{domestic} risk-free bonds (in units of the domestic currency), the family $\bigl\{\bigl(B^{\delta}(t,T)\bigr)_{t\in[0,T]},T\geq0\bigr\}$ can be thought of as representing prices of zero-coupon bonds of a \emph{foreign} ``risky'' economy, expressed in units of the foreign currency.

According to this foreign exchange analogy, one is naturally led to look at the ratio
\begin{equation}	\label{eq:fwd_premium}
R^{\delta}(t,T):=\frac{B(t,T)}{B^{\delta}(t,T)},
\end{equation}
for $t\leq T$ and $T\geq0$, where $B^{\delta}(t,T)$ ($B(t,T)$, resp.) has here to be seen as the \emph{discount factor} for the foreign (domestic, resp.) economy\footnote{Indeed, as explained in~\cite[Section 4.2.1]{MusRut}, bond prices are expressed in units of the respective currencies, while discount factors are simply the corresponding real numbers.}. Note also that $R^{\delta}(T,T)=1$, for all $T\geq0$. 
Following the presentation in~\cite[Section 4.2.1]{MusRut}, the quantity $R^{\delta}(t,T)$ corresponds to the \emph{forward exchange premium} between the domestic and the foreign currency over the time interval $[t,T]$.  Indeed, in standard foreign exchange markets, there is the following no arbitrage relation between the spot exchange rate $Q_t$ (domestic price 
of one unit of the foreign currency) and the forward exchange rate $F(t,T)$ (forward price in domestic currency of one unit of the foreign currency paid at time $T$):
\[
\frac{Q_t}{F(t,T)}=\frac{B(t,T)}{B^{\delta}(t,T)}=R^{\delta}(t,T).
\]

The multiplicative spread $S^{\delta}(t,T)$ introduced in~\eqref{eq:multspread} corresponds now to
\[
S^{\delta}(t,T)
= \frac{1+\delta L_t(T,T+\delta)}{1+\delta L^D_t(T,T+\delta)}
=\frac{B^{\delta}(t,T)}{B(t,T)}\frac{B(t,T+\delta)}{B^{\delta}(t,T+\delta)}
=\frac{R^{\delta}(t,T+\delta)}{R^{\delta}(t,T)},
\]
for all $t\leq T$ and $T\geq0$, while the spot multiplicative spread is simply given by
\[ 
S^{\delta}(T,T)=R^{\delta}(T,T+\delta),
\] 
for all $T\geq0$. In particular, note that $R^{\delta}(T,T+\delta)$ corresponds exactly to the quantity $Q^{\delta}_T$ considered in Section~\ref{subsec:model} which thus has the interpretation of a foreign exchange premium over $[T,T+\delta]$.


Since Libor rates reflect the overall credit risk of the Libor panel, the exchange rate premium $R^{\delta}(t,T+\delta)$ can be seen as a market valuation (at time $t$) of the riskiness of the foreign economy, i.e., of the credit and liquidity quality of the current Libor panel over the period $[t,T+\delta]$.
Moreover, according to the same interpretation, the quantity 
\[
S^{\delta}(t,T)=R^{\delta}(t,T+\delta)/R^{\delta}(t,T)=\mathbb{E}_{\mathbb{Q}^T}[R(T,T+\delta)|\mathcal{F}_t]
\]
is thus an expectation of the riskiness of the future Libor panel over the future time period $[T,T+\delta]$, as seen from the market at time $t$ (calculated as the relative riskiness of the current Libor panel over the period $[t,T+\delta]$ relative to the one over $[t,T]$).
For instance, a large value of $S^{\delta}(t,T)$ would mean that the market anticipates a worsening of the credit quality of the Libor panel on $[T,T+\delta]$ compared to the credit quality on $[t,T]$ as seen at time $t$.

In this sense, the multiplicative spread $S^{\delta}(t,T)$ is a rather natural quantity to model in a multiple curve setting, because it represents the market's expectation at time $t$ (being computed from financial instruments traded at date $t$) of the credit and liquidity quality of the Libor panel over $[T,T+\delta]$.

\section{Local independence and semimartingale decomposition}	\label{appendix:local_ind}

In this section, we let $(X,Y)$ be a general It\^o-semimartingale taking values in $\RR^{d+n}$ and denote by $\Psi^{X,Y}$ its local exponent and by $\Psi^X$ and $\Psi^Y$ the local exponents of $X$ and $Y$, respectively, and let $\mathcal{U}^{X,Y}$ be defined as in Definition~\ref{def:expcomp}.
In view of~\cite[Lemma A.11]{KK13}, the following definition is equivalent to the notion of local independence as given in~\cite[Definition A.10]{KK13}.

\begin{definition}	\label{def:local_independence}
We say that $X$ and $Y$ are \emph{locally independent} if, outside a $d\QQ\otimes dt$-null set, it holds that
\[
\Psi_t^{X,Y}(u_t,v_t)(\omega)=\Psi_t^X(u_t)(\omega)+\Psi_t^Y(v_t)(\omega),
\qquad\text{ for all }(u,v)\in\mathcal{U}^{(X,Y)}.
\]
\end{definition}

Following~\cite[Appendix A.3]{KK13}, let us recall the notion of semimartingale decomposition of $Y$ relative to $X$. We denote by $c^{Y,X}$ and $c^X$ the second local characteristic of $(Y,X)$ and $X$, respectively, and by $K^{Y,X}$ and $K^X$ the third local characteristic of $(Y,X)$ and $X$, respectively. Denote also by $\mu^{Y,X}$ the jump measure of $(Y,X)$. Supposing that $1\in\mathcal{U}^Y$ (i.e., $Y$ is exponentially special, see Proposition~\ref{prop:expcomp}), let
\begin{equation}	\label{eq:dep_part}
Y^{\parallel,i} 
:= \log \mathcal{E}\left(\int_0^{\cdot}\bigl(c_t^{Y^i,X}(c_t^X)^{-1}\bigr)dX^c_t
+\int_0^{\cdot}\int(e^{y^i}-1)\ind_{\{x\neq0\}}\bigl(\mu^{Y,X}(dy,dx,dt)-K^{Y,X}_t(dy,dx)dt\bigr)\right),
\end{equation}
for $i=1,\ldots,n$, where $(c^X)^{-1}$ denotes the pseudoinverse of the matrix $c^X$ and $X^c$ is the continuous local martingale part of $X$ (see~\cite[Proposition I.4.27]{jashi03}). We call $Y^{\parallel}:=(Y^{\parallel,1},\ldots,Y^{\parallel,n})^{\top}$ the \emph{dependent part of $Y$ relative to $X$} and $Y^{\perp}:=Y-Y^{\parallel}$ the \emph{independent part of $Y$ relative to $X$}. The following lemma corresponds to~\cite[Lemma A.22 and Lemma A.23]{KK13}.

\begin{lemma}	\label{lem:local_independence}
Let $(X,Y)$ be an $\RR^{d\times n}$-valued It\^o-semimartingale such that $1\in\mathcal{U}^Y$. Then the following hold:
\begin{enumerate}
\item $Y\mapsto Y^{\parallel}$ is a projection, in the sense that $(Y^{\parallel})^{\parallel}=Y^{\parallel}$;
\item if $Z$ is an It\^o-semimartingale locally independent of $X$, then it holds that $(Z+Y)^{\parallel}=Y^{\parallel}$;
\item $\exp(Y^{\parallel,i})$ is a local martingale, for all $i=1,\ldots,n$;
\item $Y^{\perp}$ and $(Y^{\parallel},X)$ are locally independent semimartingales.
\end{enumerate}
\end{lemma}

\section{Proofs of the results of Section~\ref{sec:construction}}	\label{app:proof}

\subsection*{Proof of Proposition \ref{prop:Lip}}

Let us fix any $i\in\{0,1,\ldots,m\}$.
By the same argument as in~\cite[Corollary 5.12]{F:01}, it can be shown that  $\kappa_4(H^{\lambda}_{m+1}) \subseteq H^{\lambda,0}_1 $ and $\kappa_j^i(H^{\lambda}_{m+1}) \subseteq H^{\lambda,0}_1 $ for $j=1,2$.
In the sequel, $C$ will always denote a positive constant which can vary from line to line. 
The following estimates can be derived similarly as in the proof of~\cite[Proposition 3.2]{FTT:10} (to which we refer the reader for more details), by relying on Assumption~\ref{ass:HJMMcond}, the H\"older inequality and on~\cite[Theorem 2.1]{FTT:10}. 

Concerning $\kappa_3^i(h)$ and $\kappa_5^i(h)$, we have that, for all $h \in H_{m+1}^{\lambda}$, $\xi \in \mathbb{R}^d$ and $s \in \mathbb{R}_+$,
\begin{equation*}
|(\zeta^{i0}(h)(s))^{\top}\xi|\leq C \| \zeta^{i0}(h)\|_{\lambda,d}\| \xi\|_d
\qquad\text{and}\qquad
|(\zeta^0(h)(s))^{\top}\xi| \leq C\| \zeta^0(h)\|_{\lambda,d}\| \xi\|_d,
\end{equation*}
and, for all $\xi \in \mathbb{R}^d$, $\xihat\in\RR^n$ and $s \in \mathbb{R}_+$,
\begin{align*}
\bigl|e^{u_i^{\top} \xihat+ (Z^{i0}(h)(s))^{\top}\xi}-1\bigr|&\leq Ce^{\|u_i\|_n\|\xihat\|_n+(C_0+C_i) \|\xi\|_d}(\|u_i\|_n\|\xihat\|_n+\| \zeta^{i0}(h)\|_{\lambda,d}\| \xi\|_d),\\
\bigl|e^{-(Z^0 (h)(s))^{\top}\xi}-1\bigr|&\leq Ce^{C_0 \|\xi\|_d}\| \zeta^0(h)\|_{\lambda,d}\| \xi\|_d.\\
\end{align*}
These estimates show together with~\eqref{eq:ass4} and~\eqref{eq:ass5} that $\lim_{s \to \infty}\kappa_3^i(h)(s)=0$ and 
$\lim_{s \to \infty}\kappa_5^i(h)(s)=0$.
Moreover, we have
\begin{align*}
&\int_{\mathbb{R_+}} \left(\int \bigl((\zeta^{i0}(h)(s))^{\top}\xi\bigr)^2\bigl(e^{u_i^{\top} \xihat+ (Z^{i0}(h)(s))^{\top}\xi}\bigr)K^{\Yhat,X}(d\xihat, d\xi)\right)^2 e^{\lambda s} ds
\leq C (M_0+M_i)^4K_i^2,\\
&\int_{\mathbb{R_+}} \left( \int \bigl((\zeta^0(h)(s))^{\top}\xi\bigr)^2 e^{-(Z^0 (h)(s))^{\top}\xi}F(d\xi)\right)^2 e^{\lambda s} ds 
\leq CM_0^4K_0^2,\\
&\int_{\mathbb{R_+}} \left(\int \frac{d}{ds}\bigl((\zeta^{i0}(h)(s))^{\top}\xi\bigr)\bigl(e^{u_i^{\top} \xihat+ (Z^{i0}(h)(s))^{\top}\xi}-1\bigr)K^{\Yhat,X}(d\xihat, d\xi)\right)^2e^{\lambda s} ds 
\leq C (M_0+M_i)^2K_i,\\
&\int_{\mathbb{R_+}} \left(\int \frac{d}{ds}\bigl(\zeta^0(h)\bigr)^{\top}\xi\bigl(e^{-(Z^0 (h))^{\top}\xi}-1\bigr)F(d\xi)\right)^2 e^{\lambda s} ds 
\leq CM_0^4K_0^2.
\end{align*}
In view of the form of $\partial_s \kappa^i_3(h)$ and $\partial_s \kappa_5(h)$ as given in \eqref{eq:derivativekappa3}-\eqref{eq:derivativekappa5}, this implies that $\kappa_j^i(H^{\lambda}_{m+1}) \subseteq H^{\lambda,0}_1$ for $j=3,5$. We have thus shown that $\kappa^{i}(H^{\lambda}_{m+1}) \subseteq H^{\lambda,0}_1$.

For $h_1, h_2 \in H^{\lambda}_{m+1}$ we obtain 
\begin{align*}
\| \kappa^i_1(h_1)-\kappa^i_1(h_2)\|_{\lambda,1}&\leq C(L_0+L_i)\|h_1-h_2\|_{\lambda, m+1},\\
\| \kappa^i_2(h_1)-\kappa^i_2(h_2)\|_{\lambda,1}&
\leq C(2M_i+2M_0)(L_0+L_i)\|h_1-h_2\|_{\lambda, m+1},\\
\| \kappa_4(h_1)-\kappa_4(h_2)\|_{\lambda,1}&
\leq C2M_0L_0\|h_1-h_2\|_{\lambda, m+1}.
\end{align*}
Furthermore, due to~\eqref{eq:derivativekappa3} and~\eqref{eq:derivativekappa5}, we can estimate
\begin{align*}
\| \kappa^i_3(h_1)-\kappa^i_3(h_2)\|_{\lambda,1}^2& \leq 4 (I^i_1+I^i_2+I^i_3+I^i_4),\\
\| \kappa_5(h_1)-\kappa_5(h_2)\|_{\lambda,1}^2& \leq 4 (J_1+J_2+J_3+J_4),\\
\end{align*}
where 
\begin{align*}
I^i_1&=\int_{\mathbb{R}_+}\left(\int \bigl((\zeta^{i0}(h_1)(s))^{\top}\xi\bigr)^2e^{u_i^{\top} \xihat}\bigl(e^{(Z^{i0}(h_1)(s))^{\top}\xi}-e^{(Z^{i0}(h_2)(s))^{\top}\xi}\bigr)K^{\Yhat,X}(d\xihat, d\xi)\right)^2 e^{\lambda s} ds,	\\
I^i_2&=\int_{\mathbb{R}_+}\left(\int e^{u_i^{\top} \xihat+ (Z^{i0}(h_2)(s))^{\top}\xi}
\Bigl(\bigl((\zeta^{i0}(h_1)(s))^{\top}\xi\bigr)^2-\bigl((\zeta^{i0}(h_2)(s))^{\top}\xi\bigr)^2\Bigl)
K^{\Yhat,X}(d\xihat, d\xi)\right)^2e^{\lambda s} ds,	\\
I^i_3&=\int_{\mathbb{R}_+}\left(\int \frac{d}{ds}\bigl(\zeta^{i0}(h_1)(s)\bigr)^{\top}\xi e^{u_i^{\top} \xihat}\bigl(e^{(Z^{i0}(h_1)(s))^{\top}\xi}-e^{(Z^{i0}(h_2)(s))^{\top}\xi}\bigr)K^{\Yhat,X}(d\xihat, d\xi)\right)e^{\lambda s} ds,\\
I^i_4&=\int_{\mathbb{R}_+}\left(\int \Bigl(e^{u_i^{\top} \xihat+ (Z^{i0}(h_2)(s))^{\top}\xi}-1\Bigr)
\left(\frac{d}{ds}\bigl(\zeta^{i0}(h_1)(s)\bigr)^{\top}\xi-\frac{d}{ds}(\zeta^{i0}(h_2)(s))^{\top}\xi\right)K^{\Yhat,X}(d\xihat, d\xi)\right)^2e^{\lambda s} ds,	\\
J_1&=\int_{\mathbb{R}_+}\left(\int \bigl((\zeta^0(h_1)(s))^{\top}\xi\bigr)^2 \bigl(e^{-(Z^0 (h_1)(s))^{\top}\xi}-e^{-(Z^0 (h_2)(s))^{\top}\xi}\bigr)F(d\xi)\right)^2e^{\lambda s} ds,\\
J_2&=\int_{\mathbb{R}_+}\left(\int  e^{-(Z^0 (h_2)(s))^{\top}\xi} \Bigl(\bigl((\zeta^0(h_1)(s))^{\top}\xi\bigr)^2-\bigl((\zeta^0(h_2)(s))^{\top}\xi\bigr)^2\Bigr)F(d\xi)\right)^2 e^{\lambda s} ds\\
J_3&=\int_{\mathbb{R}_+}\left(\int \frac{d}{ds}\bigl(\zeta^0(h_1)(s)\bigr)^{\top}\xi\bigl(e^{-(Z^0 (h_1)(s))^{\top}\xi}-e^{-(Z^0 (h_2)(s))^{\top}\xi}\bigr)F(d\xi)\right)^2e^{\lambda s} ds,\\
J_4&=\int_{\mathbb{R}_+}\left(\int \bigl(e^{-(Z^0 (h_2)(s))^{\top}\xi}-1\bigr)\left(\frac{d}{ds}\bigl(\zeta^0(h_1)(s)\bigr)^{\top}\xi-\frac{d}{ds}\bigl(\zeta^0(h_2)(s)\bigr)^{\top}\xi\right)F(d\xi)\right)^2e^{\lambda s} ds,
\end{align*}
We get for all $\xi \in \mathbb{R}^d$, $s \in \mathbb{R}_+$
\begin{align*}
\bigl|e^{(Z^{i0}(h_1)(s))^{\top}\xi}-e^{(Z^{i0}(h_2(s)))^{\top}\xi}\bigr|&\leq Ce^{(C_0+C_i)\| \xi\|_d}\bigl(
\|\zeta^{i}(h_1)-\zeta^{i}(h_2)\|_{\lambda,d}+\|\zeta^{0}(h_1)-\zeta^{0}(h_2)\|_{\lambda,d}\bigr)\| \xi\|_d,\\
\bigl|e^{-(Z^0 (h_1))^{\top}\xi}-e^{-(Z^0 (h_2))^{\top}\xi}\bigr| &\leq  Ce^{C_0\| \xi\|_d}\|\zeta^{0}(h_1)-\zeta^{0}(h_2)\|_{\lambda,d}\| \xi\|_d.
 \end{align*}
 Therefore,
 \begin{align*}
I^i_1&\leq C\int_{\mathbb{R}_+}\bigg(\int \bigl((\zeta^{i0}(h_1)(s))^{\top}\xi\bigr)^2e^{\|u_i\|_n \|\xihat\|_n+(C_0+C_i)\| \xi\|_d}\\
&\times\bigl(
\|\zeta^{i}(h_1)-\zeta^{i}(h_2)\|_{\delta,d}+\|\zeta^{0}(h_1)-\zeta^{0}(h_2)\|_{\lambda,d}\bigr)\| \xi\|_d
 K^{\Yhat,X}(d\xihat, d\xi)\bigg)^2 e^{\lambda s} ds,\\
 & \leq CK^2_i(M_i^4+M_0^4)(L^2_i+L_0^2)\|h_1-h_2\|^2_{\lambda,m+1},\\
 J_1&\leq C\int_{\mathbb{R}_+}\left(\int \bigl((\zeta^0(h_1)(s))^{\top}\xi\bigr)^2 e^{C_0\| \xi\|_d}\|\zeta^{0}(h_1)-\zeta^{0}(h_2)\|_{\lambda,d}\| \xi\|_dF(d\xi)\right)^2e^{\lambda s} ds\\
 &\leq C K^2_0M_0^4L_0^2\|h_1-h_2\|^2_{\lambda,m+1}.
 \end{align*}
 Moreover, for every $s \in \mathbb{R}_+$, we obtain
 \begin{align*}
  &\int e^{\|u_i\|_n \|\xihat\|_n+ (C_0+C_i)\|\xi\|_d}
  \bigl((\zeta^{i0}(h_1)(s))^{\top}\xi+(\zeta^{i0}(h_2)(s))^{\top}\xi\bigr)^2
K^{\Yhat,X}(d\xihat, d\xi)\leq 2C(M_i^2+M_0^2)K_i,\\
&\int e^{C_0\|\xi\|_d} \bigl((\zeta^0(h_1)(s))^{\top}\xi+(\zeta^0(h_2)(s))^{\top}\xi\bigr)^2F(d\xi)\leq 2CM_0^2K_0.
 \end{align*}
 Hence,
 \begin{align*}
 I^i_2&\leq 2C(M_i+M_0)K_i\int e^{\|u_i\|_n \|\xihat\|_n+ (C_0+C_i)\|\xi\|_d}\\
 &\quad \times \int_{\mathbb{R}_+}\left(\bigl(\zeta^{i}(h_1)(s)-\zeta^{i}(h_2)(s)\bigr)^{\top}\xi+\bigl(\zeta^0(h_2)(s)-\zeta^0(h_1)(s)\bigr)^{\top}\xi\right)^2e^{\lambda s} ds K^{\Yhat,X}(d\xihat, d\xi)\\
 &\leq  2C(M_i+M_0)K_i^2 (L_0^2+L_i^2)\|h_1-h_2\|^2_{\lambda,m+1},\\
 J_2 &\leq 2CM_0K_0\int e^{C_0\|\xi\|_d}\int_{\mathbb{R}_+}\left(\bigl(\zeta^0(h_2)(s)-\zeta^0(h_1)(s)\bigr)^{\top}\xi\right)^2e^{\lambda s} ds F(d\xi)\\
 &\leq 2CM_0K_0^2 L_0^2\|h_1-h_2\|^2_{\lambda,m+1}.
 \end{align*}
 Moreover,
 \begin{align*}
 I^i_3&\leq CK_i(L^2_i+L_0^2)\|h_1-h_2\|^2_{\lambda,m+1}\int e^{\|u_i\|_n \|\xihat\|_n+(C_0+C_i)\| \xi\|_d}\int_{\mathbb{R}_+}  \left(\frac{d}{ds}\bigl(\zeta^{i0}(h_1)(s)\bigr)^{\top}\xi\right)^2
  e^{\lambda s} ds K^{\Yhat,X}(d\xihat, d\xi)\\
  &\leq  CK^2_i(L^2_i+L_0^2)(M_0+M_i)\|h_1-h_2\|^2_{\lambda,m+1},\\
 J_3 &\leq CK_0L_0^2\|h_1-h_2\|^2_{\lambda,m+1}\int e^{C_0\|\xi\|_d}\int_{\mathbb{R}_+}\left(\frac{d}{ds}\bigl(\zeta^0(h_1)(s)\bigr)^{\top}\xi\right)^2 e^{\lambda s} ds F(d\xi)\\
  &\leq  CK^2_0L_0^2M_0\|h_1-h_2\|^2_{\lambda,m+1}.
 \end{align*}
Finally, we have 
 \begin{align*}
 I^i_4&\leq C\int_{\mathbb{R}_+}\bigg(\int e^{\|u_i\|_n\|\xihat\|_n+(C_0+C_i) \|\xi\|_d}\bigl(\|u_i\|_n\|\xihat\|_n+\| \zeta^{i0}(h)\|_{\lambda,d}\|\xi\|_d\bigr)\\
&\quad \times\left(\frac{d}{ds}\bigl(\zeta^{i0}(h_1)(s)\bigr)^{\top}\xi-\frac{d}{ds}\bigl(\zeta^{i0}(h_2)(s)\bigr)^{\top}\xi\right)K^{\Yhat,X}(d\xihat, d\xi)\bigg)^2e^{\lambda s} ds\\
&\leq C(u_i+M_0+M_i)^2K_i^2(L_0+L_i)^{2}\|h_1-h_2\|^{2}_{\lambda,m+1},
\end{align*}
\begin{align*}
J_4&\leq C\int_{\mathbb{R}_+}\left(\int e^{C_0 \|\xi\|_d}\| \zeta^0(h)\|_{\lambda,d}\| \xi\|_d \left(\frac{d}{ds}\bigl(\zeta^0(h_1)(s)\bigr)^{\top}\xi-\frac{d}{ds}\bigl(\zeta^0(h_2)(s)\bigr)^{\top}\xi\right)F(d\xi)\right)^2e^{\lambda s} ds,\\
 &\leq CM_0^2K_0^2L_0^2\|h_1-h_2\|^{2}_{\lambda,m+1}.
 \end{align*}
Summing up, we have shown that there exist constants $Q_i>0$ such that condition~\eqref{eq:kappaLip} is satisfied for all $h_1,h_2 \in H_{m+1}^{\lambda}$.

\subsection*{Proof of Proposition~\ref{prop:existenceY}}

For every $(\omega,\omega')\in\widetilde{\Omega}$ and $t\geq0$, let us define $Y^{\perp}_t(\omega,\omega') := y_0 + J_t(\omega')$, for some starting value $y_0\in\RR_+$. Clearly, $Y^{\perp}$ is a pure jump $(\cG_t)$-adapted process.
In order to prove the existence of a probability measure $\widetilde{\mathbb{Q}}$ such that the jump measure of $Y^{\perp}$ with respect to the two filtrations $(\cG_t)_{t\geq0}$ and $(\overline{\cG}_t)_{t\geq0}$ is given by $K_t\bigl(\omega, Y^{\perp}_{t-}(\omega,\omega'),d\xi\bigr)dt$ and $\widetilde{\mathbb{Q}}|_{\mathcal{F}}=\mathbb{Q}$ holds true, we shall rely on~\cite[Theorem 3.6]{jacod7475}. 

For all $\omega\in\Omega$, $y\in\RR_+$ and $t\geq0$, let us first extend the definition of $K_t(\omega,y, d\xi)$ as of~\eqref{eq:momentproblem}-\eqref{eq:mart_cond} to $y\in\RR_-$ by requiring that it is supported on $[-|y|, \infty)$ and by setting $p_t(y,\omega)=p_t(-y, \omega)$ for $y\in\RR_-$. Due to Assumption~\ref{ass:ex_sol_moment}, the measure $K_t(\omega,Y^{\perp}_{t-}(\omega,\omega'), d\xi)dt$ defined via the moment problem~\eqref{eq:momentproblem}-\eqref{eq:mart_cond} is a positive random measure on $\mathbb{R}_+ \times \mathbb{R}$.
Since $Y^{\perp}$ is c\`adl\`ag and $(\cG_t)$-adapted and since $\bigl(p_t(\cdot,y)\bigr)_{t\geq0}$ is $(\cF_t)$-predictable and depends in a measurable way on $y$,  the process $p_t\bigl(\omega, Y^{\perp}_{t-}(\omega,\omega')\bigr)$ is $(\cG_t)$-predictable and the same $(\cG_t)$-predictability and, hence, $(\overline{\cG}_t)$-predictability, is inherited by $K_t\bigl(\omega,Y^{\perp}_{t-}(\omega,\omega'), d\xi\bigr)$.
Let us then define the $(\overline{\cG}_t)$-predictable random measure $\nu$ by
\[
\nu(\widetilde{\omega},dt,d\xi)=\left\{ \begin{array}{ll}
K_t\bigl(\omega, Y^{\perp}_{t-}(\omega,\omega'),d\xi\bigr)dt, & t < T_{\infty};\\
0, & t \geq T_{\infty}.
\end{array} \right.
\] 
 \cite[Theorem 3.6]{jacod7475} implies that there exists a unique probability kernel $\mathbb{P}$ from $(\Omega,\mathcal{F})$ to $\mathcal{H}$, such that $\nu$ is the $(\overline{\cG}_t)$-compensator of the random measure $\mu$ associated with the jumps of $J$. On $(\widetilde{\Omega}, \mathcal{G})$ we then define the probability measure $\widetilde{\mathbb{Q}}$ by $\widetilde{\mathbb{Q}}(d\widetilde{\omega})=\mathbb{Q}(d\omega)\mathbb{P}(\omega,d\omega')$, whose restriction to $\mathcal{F}$ is equal to $\mathbb{Q}$.
Moreover, since $Y^{\perp}$ is $(\cG_t)$-adapted and $K_t\bigl(\omega,Y^{\perp}_{t-}(\omega,\omega'), d\xi\bigr)$ is $(\cG_t)$-predictable, the random measure $\nu$ is also the $(\cG_t)$-compensator of the jump measure of $J$.
Since, for every $\omega\in\Omega$, $y\in\RR_+$ and $t\geq0$, the measure $K_t(\omega,y,d\xi)$ is supported by $[-y,\infty)$, the process $Y^{\perp}$ takes values in $\RR_+$. 

It remains to show that $Y^{\perp}$ is of finite activity or, equivalently, that $T_{\infty}=\infty$ $\widetilde{\mathbb{Q}}$-a.s. Since $g_{m+1}(\xi)\geq1$ for all $\xi\in\RR$ and due to condition \eqref{eq:mart_cond}, it holds that
\begin{align*}
\EE^{\widetilde{\QQ}}\bigl[\mu\bigl([0,T]\times\mathbb{R}\bigr)\bigr]
&=\EE^{\widetilde{\QQ}}\bigl[\nu\bigl([0,T]\times\mathbb{R}\bigr)\bigr]
=\EE^{\widetilde{\QQ}}\left[\int_0^T
K_t\bigl(\omega, Y_{t-}(\omega,\omega'), \mathbb{R}\bigr)dt\right]\\
&\leq \EE^{\QQtilde}\left[\int_0^T\int g_{m+1}(\xi)K_t\bigl(\omega, Y_{t-}(\omega,\omega'),d\xi\bigr)dt\right]
= \EE^{\QQtilde}\left[\int_0^Tp^{m+1}_t\bigl(\omega,Y_{t-}(\omega,\omega')\bigr)dt\right] \leq \bar{H}T,
\end{align*}
due to the uniform boundedness of the processes $\{\bigl(p^{m+1}_t(\cdot,y)\bigr)_{t\geq0};y\in\RR_+\}$.
This implies that $\mu\bigl([0,T]\times \mathbb{R}\bigr)<\infty$, $\widetilde{\mathbb{Q}}$-a.s. for all $T \geq 0$ and, hence, $\widetilde{\mathbb{Q}}[T_{\infty}<\infty]=0$.

\subsection*{Proof of Lemma~\ref{lemma:immersion}}
(i):
for $t\geq0$, let $H$ be a bounded $\cH_t$-measurable random variable, $F$ a bounded $\cF_t$-measurable random variable and $A\in\cF_{\infty}$. As can be deduced from the proof of the previous proposition, the $\mathcal{F}_{\infty}$-conditional law of $\bigl(\omega'(s)\bigr)_{s\in[0,t]}$ under $\QQtilde$ is $\cF_t$-measurable (compare also with~\cite[part (iv) of Theorem 5.1]{FOS:11}). In particular, this means that $\EE^{\QQtilde}[H|\cF_{\infty}]=\EE^{\QQtilde}[H|\cF_t]$. In turn, this implies that
\[	\begin{aligned}
\EE^{\QQtilde}[FH\ind_A]
= \EE^{\QQtilde}[F\,\EE^{\QQtilde}[H|\cF_{\infty}]\ind_A]
= \EE^{\QQtilde}[F\,\EE^{\QQtilde}[H|\cF_t]\ind_A]
= \EE^{\QQtilde}[\EE^{\QQtilde}[FH|\cF_t]\ind_A].
\end{aligned}	\]
By a monotone class argument, this means that $\EE^{\QQtilde}[G|\cF_{\infty}]=\EE^{\QQtilde}[G|\cF_t]$ for every bounded $\cG_t$-measurable random variable $G$.
It is well-known (see e.g. \cite[Proposition 5.9.1.1]{JYC}) that the latter property is equivalent to the fact that all $\bigl(\QQtilde,(\cF_t)_{t\geq0}\bigr)$-martingales are also $\bigl(\QQtilde,(\cG_t)_{t\geq0}\bigr)$martingales. Since $\QQtilde|_{\cF_{\infty}}=\QQ$, this implies that all $\bigl(\QQ,(\cF_t)_{t\geq0}\bigr)$-local martingales are also $\bigl(\QQtilde,(\cG_t)_{t\geq0}\bigr)$-local martingales. As a consequence, every $\bigl(\QQ,(\cF_t)_{t\geq0}\bigr)$-semimartingale is also a $\bigl(\QQtilde,(\cG_t)_{t\geq0}\bigr)$-semimartingale. Moreover, since semimartingale characteristics can be characterized in terms of local martingales (see e.g. \cite[Theorem II.2.21]{jashi03}), this implies that $(X,\Yhat)$ is a semimartingale with respect to $(\QQtilde,(\cG_t)_{t\geq0})$ with unchanged characteristics.

(ii): since $Y^{\perp}$ is a pure jump process, in order to prove its local independence with respect to $(X,\Yhat)$, it suffices to show that $Y^{\perp}$ and $(X,\Yhat)$ do never jump together. In view of \eqref{eq:dep_part}, this reduces to show that $\QQtilde\bigl(\exists\, t>0 | \Delta Y^{\perp}_t\neq0\text{ and }\Delta X_t\neq0\bigr)=0$. Let $\mathfrak{T}$ be the set of jump times of $X$. Since $X$ is c\`adl\`ag, the set $\mathfrak{T}$ is countable (see e.g. \cite[Proposition I.1.32]{jashi03}) and, similarly as in \cite[Theorem 4.7]{KK13},
\[
\QQtilde\bigl(\exists\, t>0 | \Delta Y^{\perp}_t\neq0\text{ and }\Delta X_t\neq0\bigr)
\leq \EE^{\QQtilde}\left[\sum_{t\in\mathfrak{T}}\ind_{\{\Delta Y_t^{\perp}\neq0\}}\right]
= \EE^{\QQtilde}\left[\sum_{t\in\mathfrak{T}}\EE^{\QQtilde}[\ind_{\{\Delta Y_t^{\perp}\neq0\}}|\cF_{\infty}]\right]=0,
\]
where the last equality follows from the fact that $\EE^{\QQtilde}[\ind_{\{\Delta Y_t^{\perp}\neq0\}}|\cF_{\infty}]=0$ for all $t>0$, since, due to Proposition \ref{prop:existenceY}, the jump measure of $Y^{\perp}$ with respect to the larger filtration $(\overline{\cG}_t)_{t\geq0}$ (which satisfies $\overline{\cG}_0=\cF_{\infty}\otimes\{\emptyset,\Omega'\}$) is absolutely continuous with respect to the Lebesgue measure, so that $Y^{\perp}$ does not have any fixed time of discontinuity (see e.g. \cite[Lemma II.2.54]{jashi03}).

In order to prove the last assertion, note that condition \eqref{eq:mart_cond} implies that condition $I(0,1)$ from \cite{KS:02} is satisfied, since, for all $y\in\RR_+$, $T\geq0$ and $i=1,\ldots,m$,
\begin{gather*}
\sup_{t\in[0,T]}\EE^{\QQtilde}\left[\exp\left(\int_0^t\int\bigl(e^{u_i\xi}(u_i\xi-1)+1\bigr)K_s\bigl(\omega,Y^{\perp}_{s-}(\omega,\omega'),d\xi\bigr)ds\right)\right]	\\
\leq
\sup_{t\in[0,T]}\EE^{\QQtilde}\left[\exp\left((1+u_m)\int_0^t\int g_{m+1}(\xi)K_s\bigl(\omega,Y^{\perp}_{s-}(\omega,\omega'),d\xi\bigr)ds\right)\right] \\
=
\sup_{t\in[0,T]}\EE^{\QQtilde}\left[\exp\left((1+u_m)\int_0^tp^{m+1}_s\bigl(\omega,Y^{\perp}_{s-}(\omega,\omega')\bigr)
dt\right)\right] 
\leq e^{(1+u_m)T\Hline}<\infty.
\end{gather*}
Moreover, condition \eqref{eq:mart_cond} can be easily shown to imply that $\int_0^T\!\int |\xi e^{u_i\xi}-\xi|K_t\bigl(\omega,Y^{\perp}_{t-}(\omega,\omega'),d\xi\bigr)dt$ is $\QQtilde$-a.s. finite for all $T\geq0$.
Hence, \cite[Theorem 3.2]{KS:02} implies that $\bigl(\exp(u_iY^{\perp}_t-\int_0^t\Psi^{Y^{\perp}}_s(u_i)ds)\bigr)_{t\in[0,T]}$ is a uniformly integrable $\bigl(\QQtilde,(\overline{\cG}_t)_{t\geq0}\bigr)$-martingale, for all $i=1,\ldots,m$. In turn, since $T\geq0$ is arbitrary, this proves the $\bigl(\QQtilde,(\overline{\cG}_t)_{t\geq0}\bigr)$-martingale property of $\bigl(\exp(u_iY^{\perp}_t-\int_0^t\Psi^{Y^{\perp}}_s(u_i)ds)\bigr)_{t\geq0}$, for all $i=1,\ldots,m$.
Finally, since the latter process is $(\cG_t)$-adapted, it is also a martingale in the smaller filtration $(\cG_t)_{t\geq0}$.

\subsection*{Proof of Theorem~\ref{thm:constr_final}}
Due to Lemma \ref{lemma:immersion}, the local exponent of $\Yhat$ with respect to the extended filtered probability space $(\Omega,\cG,(\cG_t)_{t\geq0},\QQtilde)$ is still given by $\Psi^{\Yhat}$ and $Y^{\parallel}=(\Yhat)^{\parallel}=\Yhat$. Since $Y^{\perp}$ and $Y-Y^{\perp}=\Yhat$ are locally independent (see Lemma \ref{lemma:immersion}), the consistency condition \eqref{eq:consistencyspreadQ} directly follows from condition \eqref{eq:momentproblem}.

In order to prove the martingale property of the process given in equation \eqref{eq:martspreadQ}, note first that the $\bigl(\QQtilde,(\cG_t)_{t\geq0}\bigr)$-martingale property of the process $\Bigl(\exp\bigl(u_iY^{\perp}_t-\int_0^t\Psi^{Y^{\perp}}_s(u_i)ds\bigr)\Bigr)_{t\geq0}$ (see Lemma \ref{lemma:immersion}), condition (iv) in Definition~\ref{def:BB} and the local independence of $Y^{\perp}$ and $(X,\Yhat)$ on $(\Omega,\cG,(\cG_t)_{t\geq0},\QQtilde)$, implies that the process given in equation \eqref{eq:martspreadQ} is a $\bigl(\QQtilde,(\cG_t)_{t\geq0}\bigr)$-local martingale, for all $i=1,\ldots,m$. 
Being a non-negative local martingale, it is also a supermartingale by Fatou's lemma. Hence, to establish the true martingale property, it suffices to observe that, since $\overline{\cG}_0=\cF_{\infty}\otimes\{\emptyset,\Omega'\}$, it holds that, for all $T\geq0$ and $i=1,\ldots,m$,
\begin{align*}
&\EE^{\QQtilde}\left[\exp\left(u_iY_T+\int_0^T\bigl(\Sigma^i_s(T)-\widetilde{\Sigma}_s(T)\bigr)dX_s
-\int_0^T\Psi^{Y,X}_s\left(\bigl(u_i,\Sigma^{i\top}_s(T)-\widetilde{\Sigma}^{\top}_s(T)\bigr)^{\top}\right)ds\right)\right]	\\
&= \EE^{\QQtilde}\left[\exp\left(u_i\Yhat_T+\int_0^T\bigl(\Sigma^i_s(T)-\widetilde{\Sigma}_s(T)\bigr)dX_s
-\int_0^T\Psi^{\Yhat,X}_s\left(\bigl(u_i,\Sigma^{i\top}_s(T)-\widetilde{\Sigma}^{\top}_s(T)\bigr)^{\top}\right)ds\right)\right. \\
&\qquad\left.\EE^{\QQtilde}\biggl[\exp\left(u_iY_T^{\perp}-\int_0^T\Psi^{Y^{\perp}}_s\left(u_i\right)ds\right)\,\Bigr|\,\overline{\cG}_0\biggr]\right] \\
&= \EE^{\QQtilde}\left[\exp\left(u_i\Yhat_T+\int_0^T\bigl(\Sigma^i_s(T)-\widetilde{\Sigma}_s(T)\bigr)dX_s-\int_0^T\Psi^{\Yhat,X}_s\left(\bigl(u_i,\Sigma^{i\top}_s(T)-\widetilde{\Sigma}^{\top}_s(T)\bigr)^{\top}\right)ds\right)\right]e^{u_iY^{\perp}_0} \\
&= \EE^{\QQ}\bigl[\exp(u_i Y_0)\bigr],
\end{align*}
where in the last equality we have used the fact that $\QQtilde|_{\cF}=\QQ$ and the $\bigl(\QQ,(\cF_t)_{t\geq0}\bigr)$-martingale property of \eqref{eq:martYX}.

\bibliographystyle{abbrv}

\bibliography{biblio150505}

\end{document}